\documentclass[10pt,journal,compsoc]{IEEEtran}

\usepackage[cmex10]{amsmath}
\usepackage{algorithmic}
\usepackage{algorithm}
\usepackage{array}

\usepackage{amsmath}
\usepackage{amsthm}
\usepackage{relsize}
\usepackage{fixltx2e}  

\usepackage{graphicx}
\usepackage{caption}
\usepackage{subcaption}
\usepackage[bookmarks=false]{hyperref}

\usepackage{stfloats}

\usepackage{epsfig}
\usepackage{epstopdf}

\usepackage{hhline}
\ifCLASSOPTIONcompsoc
  \usepackage[nocompress]{cite}
\else
   normal IEEE
\usepackage{cite}
\fi

\newtheorem{lemma}{Lemma}
\newtheorem{theorem}{Theorem}

\IEEEoverridecommandlockouts
\begin{document}

\title{A Robust Advantaged Node Placement Strategy for Sparse Network Graphs}

\author{ Kai~Ding, ~\IEEEmembership{Student Member,~IEEE, } 
		 Homayoun~Yousefi'zadeh, ~\IEEEmembership{Senior~Member,~IEEE,}\\ 
		 and
		 Faryar~Jabbari, ~\IEEEmembership{Senior~Member,~IEEE}
\IEEEcompsocitemizethanks{\IEEEcompsocthanksitem This work was supported in part by the DARPA GRAPHS program Award N66001-14-1-4061. An earlier version of this work appears in \cite{GDO}. \protect\\ 
\IEEEcompsocthanksitem 	Kai Ding and Faryar Jabbari are with the Department of Mechanical and Aerospace Engineering, UCI . Homayoun Yousefi'zadeh is with Center for Pervasive Communications and Computing, UCI.  \texttt{email: [kaid1,hyousefi,fjabbari]@uci.edu}
		}}

\IEEEtitleabstractindextext{
\begin{abstract}
Establishing robust connectivity in heterogeneous networks (HetNets) is an important yet challenging problem. 
For a HetNet accommodating a large number of nodes, establishing 
perturbation-invulnerable connectivity is of utmost importance. 
This paper provides a robust advantaged node placement strategy best suited 
for sparse network graphs. 
In order to offer connectivity robustness, this paper models 
the communication range of an advantaged node with a hexagon embedded 
within a circle representing the physical range of a node. 
Consequently, the proposed node placement method of this paper is based on a so-called hexagonal coordinate system (HCS) in which we develop an extended algebra. 
We formulate a class of geometric distance optimization
 problems aiming at establishing robust connectivity of a graph of multiple clusters of nodes.
 After showing that our formulated problem is NP-hard, we utilize HCS to efficiently solve 
 an approximation of the problem.  
First, we show that our solution closely approximates an exhaustive search solution approach 
for the originally formulated NP-hard problem. Then, we illustrate its 
advantages in comparison with other alternatives  
through experimental results capturing 
advantaged node cost, runtime, and robustness characteristics.
The results show that our algorithm is most effective in sparse networks for which we derive 
classification thresholds.
\end{abstract}

}
\maketitle

\IEEEraisesectionheading{\section{Introduction}}
\IEEEPARstart Establishing connectivity in heterogeneous networks has been of high significance in the studies of HetNets in
MANETs, WSNs, and multi-facility locations \cite{MFLP}.
HetNets are typically composed of nodes with different capabilities and are formed by a collection of clusters. Generally, each cluster contains several standard nodes 
with short communication ranges and a cluster head node \cite{Survey}. The cluster head node is an advantaged node serving as the gateway of this cluster in communication with other cluster heads. 
Connectivity scenarios of multi-tier networks have found extensive applications in different disciplines including but not limited to health surveillance, environment monitoring, earthquake detection, and Internet of Things (IoT). In all these applications, a large number of low-capability standard nodes (SNs) rely on a small number of advantaged nodes (ANs) to communicate. 

Similar to literature work of \cite{Hou,Lifetime1}, 
this paper assumes HetNets are formed by SNs arranged in clusters with each cluster 
designated an AN gateway. AN gateways are assumed to have much longer communication ranges 
and able to simultaneously connect  to multiple nodes \cite{AP}. 
While the assumption guarantees intra-cluster connectivity 
and a certain length of life-time \cite{Lifetime1}, inter-cluster connectivity still needs to be 
established by placement of additional intermediate ANs. 
Lin and Xue \cite{MSTh} abstract this problem in the form of 
a \emph{Steiner minimum tree problem with minimum number 
of Steiner points and bounded edge length}. 
We refer to this algorithm as SMT not to be confused with MST used to represent 
minimum spanning trees.
Lin and Xue provide an approximation algorithm to the original NP-complete problem 
with a polynomial time complexity and a performance ratio of $5$. 
This algorithm lays the groundwork of several other approximation algorithms with smaller (better)  performance ratios \cite{Lloyd, ApproxSMT, XCheng, EBST}. 
In \cite{AP}, the authors develop a node placement algorithm for clustered ad-hoc networks 
subject to capacity constraints.  
Other related works, albeit at small scale sensor networks, include \cite{largescale,Traff,Lifetime1, Lifetime2} in which energy  
and network lifetime constraints are emphasized in node placement. 

All of the above algorithms use the Gilbert disk connectivity model
\cite{Gilbert, NLi, JPan}
representing the communication range of an AN as a circle. 
One disadvantage of this model is lack of boundary connectivity robustness 
where the distance between two centers is close to the distance threshold of
connectivity $d$.   
In such cases, a pair of connected nodes can easily become disconnected 
as the result of small position perturbations, a phenomenon occurring 
frequently and unpredictably, especially in harsh environments. 
To compensate against these cases, fault-tolerant $k$-connectivity ($k \geq 2$) node placement 
algorithms have been developed 
\cite{FTRelay, largescale, FTinHet, faulttol, FTsensor}.
By using a much larger number of ANs, these algorithms guarantee 
there are always $k$ different paths between each pair of ANs.

In addition to the disadvantage above, SMT-based methods  
are subject to a second yet major disadvantage.
Since the minimum spanning tree is formed once statically  
to represent the topology of the network graph,  
SMT-based methods do not consider the effects of changes to minimum spanning 
tree as the result of placing ANs in subsequent iterations. 
This can lead to potentially over utilizing AN resources, 
since it is possible to establish connectivity with a smaller number of ANs. 

As detailed in Section \ref{hcsSec} and Section \ref{caa}, 
this work provides a dynamic strategy for AN placement 
capable of dynamically considering the effects of changes to minimum spanning 
tree while offering robust network connectivity in the presence 
of perturbations. In essence, we seek an AN placement strategy that  
carries a certain level of robustness therein. 
To avoid the inherent problem of Gilbert disk model in boundary connectivity cases, 
we model the communication ranges of nodes as hexagons embedded within the circles representing 
the actual communication ranges of nodes. Two nodes are considered connected only when their associated  
hexagons have a common edge. 
Consequently, a pair of connected nodes actually have a margin of perturbation conserving connectivity. 
Projecting the node placement problem into HCS with integer coordinates
allows us to utilize the higher computation efficiency of HCS compared to 
a conventional Cartesian Coordinate System (CCS) in minimizing the number 
of intermediate ANs, identifying their positions, and accounting for topology perturbations.

In our work, we consider a two-tier graph of nodes in which clusters of SNs are to be connected with a minimum number of ANs. ANs are distinguished from SNs by their higher ranges of communication
and ability to simultaneously connect to a large number of standard nodes. Each cluster of SNs is assumed to be equipped with an AN allowing full connectivity of the nodes within the cluster. Multiple clusters of SNs may or may not be connected depending on their separation distance.  
It is important to note that inter-cluster connectivity as facilitated by ANs is mostly a function of distance as opposed to interference because of the much larger separation distances of ANs 
and much stronger power profiles compared to SNs.

The main contribution of our work is as follows. 
First and for the purpose of offering robustness, we introduce a hexagonal coordinate 
system and develop associated extended algebra. 
Relying on the proposed HCS, we then formulate a class of geometric distance optimization problems aiming at finding the minimum number of ANs and their positions 
to guarantee robust connectivity of a given HetNet. 
We prove that our formulated problem is 
NP-hard and offer an exhaustive search algorithm for solving this 
NP-hard problem as well as a low complexity algorithm for solving an approximation 
of this problem. 
We show our heuristic solution closely tracks 
the exhaustive search algorithm while enjoying excellent node cost, runtime, and robustness characteristics compared to other alternatives.
Our proposed approximation algorithm utilizes far fewer ANs 
than a $k$-connected network. 
This is because establishing 
a $k$-connected network requires many additional edges to a graph so as to preserve connectivity under $(k-1)$ edge or vertex cuts. Naturally, adding edges will increase the number of intermediate ANs.

The rest of the paper is organized as follows. 
Section II describes 
connectivity model. In Section III, the hexagonal 
coordinate system and the associated algebra are introduced.
Section IV describes the formulation of the connectivity problem, the proof of NP-hardness, and an exhaustive search algorithm solving the problem. Section V includes the heuristic node placement algorithm and the associated analysis. Section VI contains our experimental results. 
Finally, Section VII concludes the paper.

\section{Connectivity Model}
Based on the landmark Gilbert connectivity model \cite{Gilbert}, early connectivity models in network graphs mainly consider the distance between nodes. Later, a number of more realistic models \cite{MIMO, Doussee, Lin} were established to capture connectivity using propagation, fading, shadowing, signal-to-interference-noise ratio (SINR), symbol error rate (SER), and capacity. A review of these recent works reveals that using a distance-based connectivity 
model is justified when high power long range communication  
dominates other factors such as interference, fading, and shadowing. Accordingly, this work 
assumes that ANs are characterized by longer communication ranges, higher powers, 
and higher lifetimes compared to SNs.

A pair of nodes \{$M$, $N$\} are considered to be bi-directionally 
connected if both $M$ and $N$ located within each other's communication range .
In the definition above, the distance between nodes is a realistic
measure of connectivity because inter-cluster communication relies on LOS links
established betweem high power ANs. 
For a pair of SN and AN nodes with ranges $r$ and $d$ in radii, 
connectivity is established only when the distance between two nodes is less than or equal to 
$\min\{r, d\}=r$.  

In our model, a number of SNs form a connected cluster for which the center of geometry can be calculated. A number of these clusters in a given area compose a network topology scenario. 
The location of clusters could be random or follow some certain distribution rule 
depending on the SN deployment preference. The 3 red dots in  Fig. \ref{coord} represent  
3 SNs with communication ranges of $r$ forming a sample cluster of SNs. 
Each cluster is assumed to be 
supported by an AN gateway node. This AN is typically located at the center of geometry of the cluster in order to maximize the number 
of SNs to which it is directly connected. Alternatively, 
AN gateways may have a small displacement from the center of geometry. Nonetheless, SNs within a cluster are all connected 
to the AN gateway node and able to communicate with nodes outside of the cluster through the AN gateway node. 
Thus, the problem of global connectivity is converted to connecting individual clusters utilizing additional intermediate 
AN nodes as necessary. Based on the connectivity condition given above, one AN ought to locate within the communication range of another AN so as to establish connectivity. 
A pair of ANs with communication ranges of $d=2R$ are connected if the two circles with radii $R$ and centered around them overlap.

In order to provide a margin of robustness in presence of location perturbation,
we model the communication area of an AN by a hexagon with an edge length of $R$. 
Considering the extended range of AN compared to SN, we assume $R$ is approximately 
two orders of magnitude 
larger than $r$. Without loss of generality, the communication area of an AN can be set as a hexagon with an edge length of 
$(12n+7)r$ where $n$ is a positive integer chosen such that the expression accurately  approximates the value of $R$.
The length selection of $(12n + 7)r$ offers a couple of geometrical advantages. First, any vertex of a large hexagon overlaps with the vertex of a hexagonal cell at the same relative position. Second, center to edge distance of a hexagon is conveniently measurable by the distance measure defined in the next section. This distance relates to the minimum distance coverage by an AN and will be utilized in Section \ref{caa}.
Two ANs are then robustly connected if their associated hexagons  have a common edge.

\section{Algebra in Hexagonal Coordinate System}
\label{hcsSec}
First, the node placement problem is projected into a so-called hexagonal coordinate system. To set up the 
HCS, we have to specify the origin, axes, and coordinates. The origin is defined as the center of geometry of all clusters. From this 
origin, we start tiling the plane with hexagonal cells. These cells have an edge length equal to the communication radius of an 
SN, $r$. The first cell share the same center of geometry as the origin point with coordinates $(0, 0)$. Then, we establish the rest of the
tessellation with equal-sized hexagonal cells. Theoretically, an infinite tessellation can tile an infinite-extending plane without either overlapping 
or gaps. In practice, we stop when the area of interest is fully tiled. 
The $x$-axis goes through the origin and is perpendicular to a pair of parallel edges of the cell containing the origin. 
The $x$-axis cuts through all hexagonal cells along that direction through their center and edge. The $y$-axis is defined as the rotation 
of the $x$-axis by ${\pi}/{3}$ counter-clockwise, as shown in Fig. \ref{coord}. The $y$-axis also crosses the origin and vertically 
cuts across the edges of all cells along the way including origin.
\begin{figure}
\centering
\includegraphics[width=0.4\textwidth]{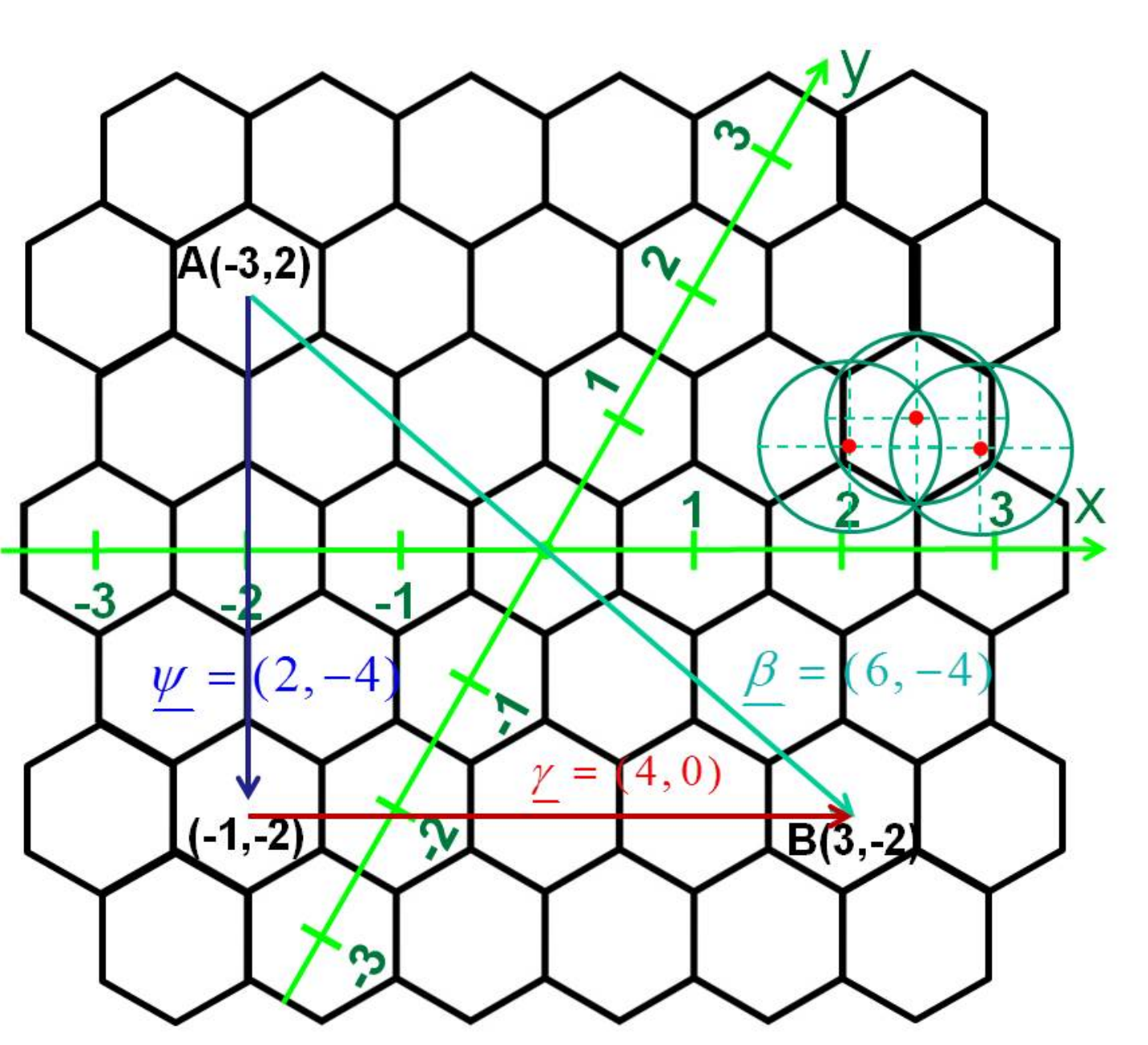}
\caption{The hexagonal coordinate system (HCS).}
\label{coord}
\end{figure}
In this coordinate system, coordinates are associated with those of a hexagonal cell unlike 
other coordinate systems such as that of \cite{Honey} in which the $x-$axis goes through the center and a cell vertex. 
Points $A$ and $B$ in Fig. \ref{coord} illustrate a pair of coordinate examples.

\subsection{Operation Definitions in HCS}
\subsubsection{Distance Measure}
Since a point in an HCS actually represents the location of a hexagonal cell, a distance measure between two hexagonal cells aims at 
counting the number of cells moving from one cell to another. The distance between point A and B in Fig. \ref{coord} serves as 
a typical example. 
For a given pair of points $M (m_1, m_2)$ and $N (n_1, n_2)$, the distance measure for the vector 
$\underline{\vartheta}=(M,N)$ is defined as follows.
\begin{equation}
\begin{array}{ll}
& |\underline{\vartheta}|=|(M, N)| =  (m_1-n_1, m_2-n_2) \\
 & \quad =  \mathbf{max}  \{|m_1-n_1|, |m_2-n_2|, |m_1-n_1+m_2-n_2|\} 
\end{array}
\label{dis}
\end{equation}
For example, $\underline{\beta}=(A,B)$ 
is a vector starting at the center of cell $A$ and ending at the center of 
cell $B$ in Fig. \ref{coord}. 
The distance between $A$ and $B$ is $6$ representing the shortest path from $A$ to $B$ covers 6 cells. 
\begin{equation*}
 \underline{\beta}=(A,B)=(6,-4),
\quad
 |\underline{\beta}|=|(A,B)|= 6
\end{equation*}
\begin{theorem}
The distance measure defined by (\ref{dis}) is a distance.
\end{theorem}

\begin{proof}
Noticing that a distance in HCS calculated through (\ref{dis}) is non-negative, it 
is left to prove the triangular inequity:
\begin{equation*}
\forall  L, M, N,   |(L, M)| + |(M, N)| \geq |(L, N)| 
\end{equation*}
We have the following three cases to consider.

\emph{Case 1} $ \text{RHS} = |l_1-n_1|$. 
In this case, we have 
\begin{equation*}
\begin{split}
\text{LHS} &  =\mathbf{max}  \{|l_1-m_1|, |l_2-m_2|, |l_1-m_1+ l_2-m_2|\} \\
 &  + \mathbf{max}  \{|m_1-n_1|, |m_2-n_2|, |m_1-n_1+m_2-n_2|\} \\
 & \geq |l_1-m_1| + |m_1-n_1|\\
 & \geq |l_1 - n_1| = \text{RHS}
\end{split}
\end{equation*}

\emph{Case 2} $ \text{RHS} = |l_2 - n_2|$.
In this case, we have
\begin{equation*}
\begin{split}
\text{LHS} &  =\mathbf{max}  \{|l_1-m_1|, |l_2-m_2|, |l_1-m_1+ l_2-m_2|\} \\
 &  + \mathbf{max}  \{|m_1-n_1|, |m_2-n_2|, |m_1-n_1+m_2-n_2|\} \\
 & \geq |l_2-m_2| + |m_2-n_2|\\
 & \geq |l_2 - n_2| = \text{RHS}
\end{split}
\end{equation*}

\emph{Case 3} $ \text{RHS} = |l_1 - n_1 + l_2 - n_2|$.
In this case, we have
\begin{equation*}
\begin{split}
\text{LHS} &  =\mathbf{max}  \{|l_1-m_1|, |l_2-m_2|, |l_1-m_1+ l_2-m_2|\} \\
 &  + \mathbf{max}  \{|m_1-n_1|, |m_2-n_2|, |m_1-n_1+m_2-n_2|\} \\
 & \geq |l_1-m_1+ l_2-m_2| + |m_1-n_1+m_2-n_2|\\
 & \geq |(l_1-m_1+ l_2-m_2) + (m_1-n_1+m_2-n_2)| \\
 & =  |l_1 - n_1 + l_2 - n_2| = \text{RHS}
\end{split}
\end{equation*}

\end{proof}

The vector addition rule in HCS follows that of Cartesian coordinate system. 
For vector $\underline{\psi}= (\psi_1, \psi_2)$ and $\underline{\gamma}= (\gamma_1, \gamma_2)$, we have 
$\underline{\psi} + \underline{\gamma} = (\psi_1+\gamma_1, \psi_2+\gamma_2)$.
\subsubsection{Inner Product}
The definition of inner production in an HCS is not the same as that of Cartesian coordinate system since the two basis vectors are not perpendicular to each other. Let's call $\underline{e}_1$ and $\underline{e}_2$ the two basis vectors in the HCS along x- and $y$-axis, respectively. 
We define the inner product as follows.
\begin{equation}
	\underline{x}\cdot \underline{y}=\underline{x}^T\Delta\underline{y}
\end{equation}
where $\underline{x} = x_1\underline{e}_1 + x_2\underline{e}_2$,	
$\underline{y} = y_1\underline{e}_1 + y_2\underline{e}_2$, 
$T$ represents the transpose operator, and $\Delta$ is a symmetric matrix defined below. 
\begin{equation}
\Delta = 
\begin{bmatrix}
	1 & \cos\dfrac{\pi}{3}\\
	\cos\dfrac{\pi}{3} & 1
\end{bmatrix}
\end{equation}
It is observed that the inner product of a pair of vectors is zero if they are perpendicular to each other, as shown by vector 
$\underline{\gamma}$ and $\underline{\psi}$ in Fig. \ref{coord}. Further, the inner product operation is commutative as $\Delta$ 
is a symmetric matrix.  
\subsection{Orientation of Distance Vector}
When dealing with the least number of ANs required to link two clusters, one has to realize that 
the maximum covered distance of an AN has its own orientation. 
When the distance vector between two clusters is closely aligned with $x$- or $y$-axis, one has to possibly use more ANs  than a case 
in which the distance vector is oriented at the direction ${\pi}/{6}$ away from each axis. In the latter case, one uses the length of 
diagonal to divide the distance and decide how many ANs are needed. Since we are mainly concerned about whether the distance 
vector is more aligned with $x$-, $y$-axis, or with the diagonals of the head and tail cluster, we take the basis axis as the reference 
point of the orientation. 
In the following subsections, we discuss a number of cases in which 
$\underline{\vartheta}$ is in different quadrants. We specify the quadrant in which 
$\underline{\vartheta}$ is located by inspecting a vector parallel to and of the same length as $\underline{\vartheta}$ 
with a starting point at the origin. 

\subsubsection{$\underline{\vartheta}$ is in the 1st or 3rd Quadrant}
As shown in Fig. \ref{13q}, the distance vector between cluster $A$ and $B$ is $(A,B)$. 
The orientation of $(A,B)$ is represented by $\theta$ which is between $(A,B)$ and $x$-axis.
\begin{figure}
\centering
\includegraphics[width= 0.4\textwidth]{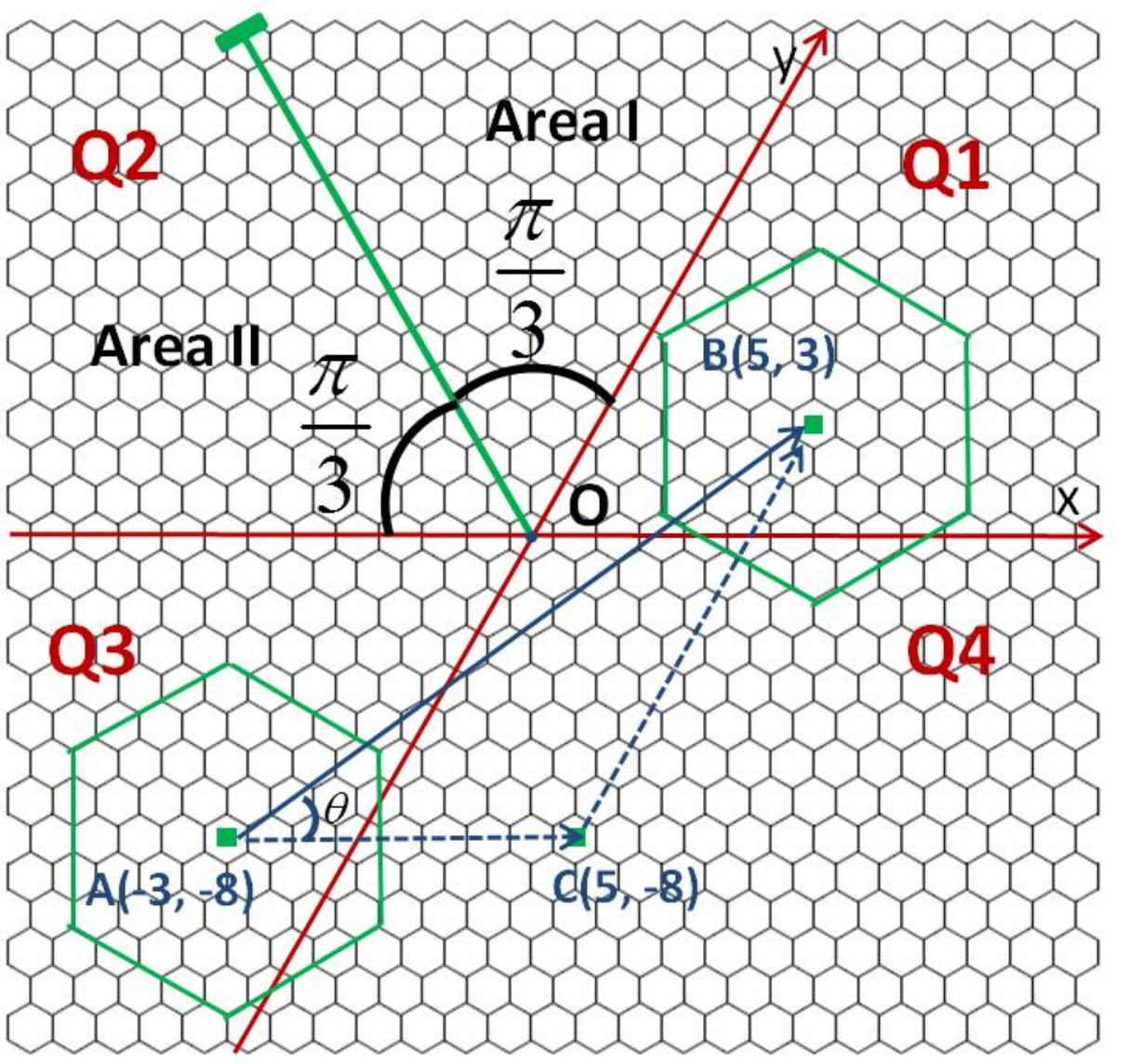}
\caption{Calculating angle $\theta$ between $\underline{\vartheta}$ and $x^{(+)}$-axis when 
$\underline{\vartheta}$ is in the 1st or 3rd quadrant. 
In the other two quadrants, the quadrant is divided into two areas  and 
$\theta$ is calculated with respect to $y^{(+)}$-axis in Area I or $x^{(-)}$-axis in Area II.}
\label{13q}
\end{figure}
\begin{figure}
\centering
\includegraphics[width=0.4\textwidth]{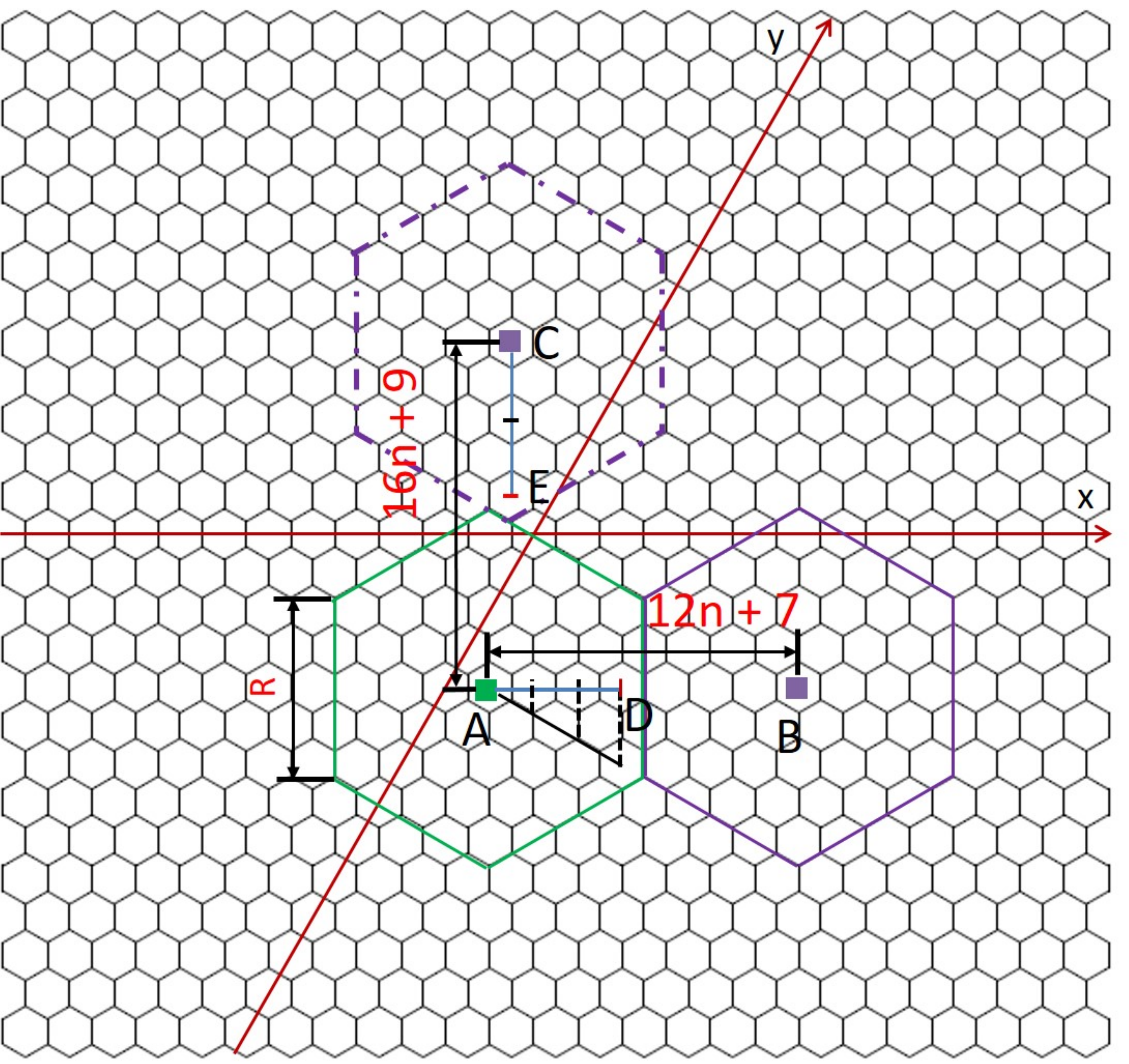}
\caption{The longest possible distance covered by an AN. }
\label{cov}
\end{figure}
In the triangle formed by $A$, $B$, and $C$, we can identify the value of $\theta$ 
from the Law of Sines, with $x=|(A,C)|$ and $y=|(C,B))|$, as 
$
\frac{1}{y}\sin{\theta} = \frac{1}{x}\sin(\frac{\pi}{3}-\theta)
$.
Hence, 
\begin{equation}
	0 < \theta = \tan^{-1}\big({\frac{\sqrt{3}\,y}{2x+y}} \big) < \frac{\pi}{3}
\label{th13}
\end{equation}
\subsubsection{$\underline{\vartheta}$ is in the 2nd or 4th Quadrant}
In HCS, quadrants 2 and 4 are larger in area than quadrants 1 and 3. The angle between $y^{(+)}$-axis and $x^{(-)}$-axis is 
${2\pi}/{3}$ which exceeds ${\pi}/{2}$. Whether we take $y^{(+)}$- or $x^{(-)}$-axis as the reference, the method of the former subsection leads to a point of discontinuity in Eq. (\ref{th13}). Therefore, we partition each quadrant into two areas, as shown in quadrant 2 of Fig. \ref{13q}. 
In Area I, the orientation of distance vector is referred to as $y^{(+)}$-axis, while in Area II, it is referred to as $x^{(-)}$-axis. 
Then, we can extract the associated equations from the Law of Sines separately.
In Area I with $|x|\leq |y|$, we have $\frac{1}{-x}\sin{\theta} =\frac{1}{y} \sin(\frac{2\pi}{3}-\theta)$.  
Therefore, 
\begin{equation}
0 < \theta = \tan^{-1}\big({\frac{-\sqrt{3}\,y}{2y+x}} \big)< \frac{\pi}{3}
\end{equation}
In Area II, $|x|\geq |y|$, we have $\frac{1}{-y}\sin{\theta} = \frac{1}{x}\sin(\frac{2\pi}{3}-\theta)$. Then,
\begin{equation}
0 < \theta = \tan^{-1}\big({\frac{-\sqrt{3}\,y}{2x+y}}\big) < \frac{\pi}{3}
\end{equation}
With distance orientation information, we are able to calculate the least number of intermediate ANs required to connect two clusters. 
That is to get the longest covering range of one AN along the direction of the distance vector and then divide the distance by the range.
The following lemma gives the possible longest distance covered by an AN with a communication range $R = (12n + 7)r$.

\begin{lemma}
\label{integer}
The longest possible distance covered by an AN is an odd integer between $(12n + 7)$ and $(16n + 9)$ corresponding to direction relative to its neighboring AN.
\end{lemma}

\begin{proof}
Fig. \ref{cov} shows two extreme cases. Segment $|(A, B)|$ is the shortest possible distance covered by adding one AN, $B$, that is connected to $A$. The direction of vector $(A, B)$ is perpendicular to the common edge. On the other hand, segment $|(A,C)|$ is the longest possible distance covered by adding one AN. Here, we calculate them separately. Recalling that the edge length of a large hexagon in Fig. \ref{cov} is $R$, that of a small hexagon cell is $r$, and $R = (12n + 7)r$, we have
\begin{equation*}
|(A, B)| = 2 \times |(A, D)| + 1 = 2\times \dfrac{R-1}{2r} + 1 = 12n + 7
\end{equation*}
\begin{equation*}
|(A, C)| = 2 \times |(C, E)| + 1 = 2 \times (2\times\dfrac{R-1}{3r}) + 1 = 16n + 9
\end{equation*}
As long as two ANs are connected and have one common edge, the distance between them in HCS is larger than the minimum case $|(A, B)|$ and less than the maximum case $|(A, C)|$. Last but not least, if two ANs have one edge in common, the distance between them is always an odd number. 
\end{proof}

\section{NP-hard Problem Statement and Exhaustive Search Algorithm}
In this section,  we prove that our node placement problem in HCS is NP-hard  
by showing that it is a reduction from Knapsack problem which is known to be NP-complete.
Then, we provide an exhaustive search algorithm to solve the problem as a comparison benchmark. 

\subsection{NP-Hard Problem Statement}

\textbf{Problem 1} (Knapsack Problem \cite{handbook}) 
Given a set of items $E = \{e_1,\cdots,e_t\}$
each with a weight $w_i$ and a profit $v_i$ where $i \in \{1,\cdots,t\}$,  
is there a way of choosing $x_i$ units of each item $e_i$ to fill the knapsack 
such that the profit of the items chosen $\sum_{i=1}^t x_i v_i$
is at least $V$ while the total weight of the items chosen 
$\sum_{i=1}^t x_i w_i$ is not exceeding $W$?\\

\noindent \textbf{Problem 2} (Node placement problem in HCS) Given $G$ pre-deployed gateway nodes with integer coordinates in an HCS and a minimum spanning tree of length $L$ formed by these nodes, can one cover the total distance of $L$ by $N$ additional intermediate nodes?

In Problem 2, covering length $L$ with $(N+G)$ ANs is equivalent to being able to find 
a connected path between any arbitrary pair of nodes where every pair of neighboring 
nodes have distances in the range $[(12n+7)r,(16n+9)r]$.

\begin{theorem}
\label{redc}
There is a polynomial time reduction from Problem 1 to Problem 2.   
\end{theorem}

\begin{proof}
Suppose set $D$ has cardinality $\mu$ and 
contains all odd integers between $(12n + 7)$ and $(16n + 9)$, i.e., 
$D = \{12n + 7, 12n + 9, \dots, 16n + 9\}$. 
We start with an instance $I$ of Problem 1 with which set $E$ with cardinality $\mu$ is associated. Then, we construct an instance $I'$ of Problem 2 with which set $D$ 
also with cardinality $\mu$ is associated.

According to Lemma \ref{integer}, each intermediate AN, based on the orientation of distance vector to its neighbor, covers a distance type $\underline{d_i}$ where $|\underline{d_i}| \in D$. 
These $\underline{d_i}$'s are items to be packed in instance $I'$. 
Let $y_i$ denote the number of ANs of type $\underline{d_i}$.
Then, the total distance $l$ covered by all intermediate ANs is expressed as
\begin{equation}
l = \sum_i y_i|\underline{d_i}|
\end{equation}
Let profit $V$ in instance $I$ be equal to $L$. 
Assuming the weight of each AN is $1$, i.e., $w_i = 1$, the total weight of all intermediate ANs amounts to the number of intermediate ANs, i.e., 
\begin{equation}
\sum_i y_i  = N
\end{equation}
Considering the statement above and assuming $W = N$, 
the process of constructing $I'$ from $I$ occurs in polynomial time.

In Problem 2, we are seeking a yes/no answer to the question  ``Can
we, by using $N$ intermediate ANs, cover a total distance of $l \geq L$?''
If the answer to Problem 1 is yes,
we can fill the knapsack such that a minimum profit of 
$V$ is reached without exceeding a maximum weight of $W$. 

Through the reduction above, it is feasible to cover a length of at least $L$ by
placing at most $N$ additional nodes. In addition, if the answer to Problem 2 is
no, which is a special instance $I$ of Problem 1 with $V = L$, $w_i = 1$, 
and $W = N$, Problem 1 will also have no answer. This implies a
polynomial reduction from Problem 1 to Problem 2.
\end{proof} 

Therefore, we conclude that Problem 2 is NP-hard. In the next subsection, we provide an exhaustive search algorithm to solve Problem 2.

\subsection{Exhaustive Search Algorithm}

Our exhaustive search algorithm uses a number of intermediate ANs and tries to rearrange their locations so as to establish global connectivity, until the smallest number of ANs that connect the entire graph is identified. We use $\kappa$ to denote the number of intermediate ANs in exhaustive search.
There is a finite set of feasible locations representing the candidate coordinates of intermediate AN locations in HCS. Generally speaking, all coordinate points of HCS except those occupied by pre-deployed clusters are feasible. The number of feasible locations $M$ is then derived as 
\begin{equation}
M = \lceil \dfrac{\Omega}{1.5\sqrt{3}\,r^2}\rceil - G 
\end{equation}
where $1.5\sqrt{3}\,r^2$ represents the area of a hexagonal cell of edge length $r$, 
$\Omega$ is the field area, and $G$ is the number of pre-deployed clusters.
Those $M$ feasible locations are then stored in an $M\times 2$ matrix $\mathbf{F}$. 

In our exhaustive search algorithm, we test all ${M \choose \kappa}$ possible combinations of $M$ coordinates in $\mathbf{F}$ and check if there is one configuration that accomplishes connectivity of all clusters. If not, we increase $\kappa$ by $1$ and repeat the same process until the least number of intermediate ANs rendering global connectivity is reached. The algorithmic pseudo code is given in Algorithm \ref{alg2}.

One may notice the considerable computational complexity of the nested 'for' loop. Given $M$ feasible AN locations, when the optimal solution is reached, say connecting the entire graph 
with $\kappa$ ANs, then the runtime of exhaustive algorithm 
is at least in the order of 
$\mathcal{O}(\sum_{i=1}^{\kappa - 1} {M \choose i})$. 
It can be shown, by Stirling's formula and binomial theorem, that the runtime is bounded 
as shown below. 
\begin{equation}
\dfrac{(M+1)^{\kappa-1}}{(\kappa-1)!}< \mathcal{O}(\sum_{i=1}^{\kappa - 1} {M \choose i}) < 2^M
\end{equation}

In practice, we are able to strategically preclude some locations that have very low possibility of accommodating an AN. For instance, an AN may not be placed too close to a pre-deployed cluster, and all ANs typically are, but not always, located inside the convex hull containing all pre-deployed clusters. With this strategy, we can reduce the size of $\mathbf{F}$ to some extent. However, to the best of our knowledge, there is no systematic strategy of reducing the number of feasible locations.   
\begin{algorithm}
	\caption{Exhaustive Search Algorithm}
	\label{alg2}
	\begin{algorithmic}[5]
		\STATE \textbf{Input:} \text{Location of pre-deployed clusters}
		\STATE \textbf{Output:} \text{Coordinates of intermediate ANs} 
		\STATE \text{	}
		\STATE \text{Establish finite HCS and} 
		\STATE \text{Put coordinates of feasible AN positions in $\mathbf{F}_{M\times2}$}
		\STATE \text{Set $\kappa = 0$} 	
		\WHILE{(Graph $\mathcal{G}$ not fully connected)}
		\STATE \text{$\kappa$ $+=$ $1$}
		\FORALL{${M \choose \kappa}$ possible combos of feasible positions} 
		\STATE \text{Place $\kappa$ ANs at these positions}
		\IF{$\mathcal{G}$ is connected}
		\STATE \text{Set coordinates of $\kappa$ ANs from $\mathbf{F}_{M\times 2}$ }
		\STATE \textbf{break}
		\ENDIF 
		\ENDFOR 
		\ENDWHILE
	\end{algorithmic}
\end{algorithm}

\section{Heuristic Connectivity Algorithm}
\label{caa}

As the node placement problem described in the previous section is NP-hard, 
it is realistic to find a heuristic near-optimal solution offering a reasonable time complexity.
Hence, this section provides a description of our heuristic connectivity algorithm 
and its complexity analysis. 

As illustrated in section II, 
we model the communication range of an AN by a hexagonal cell 
with an edge length of $(12n+7)r$ where $n$ is a positive natural number. 
A robust connectivity criterion is defined as when two large hexagons have a common edge. 
Consequently, we are dealing with the task of connecting a number of hexagons within 
the network graph by optimally placing a number of hexagons of the same size between each pair as needed. 
As this task can be abstracted as a distance optimization problem within 
the HCS, we introduce a class of geometric distance 
optimization (GDO) algorithms.

The main algorithm of interest in this paper 
is referred to as enhanced geometric distance 
optimization (EGDO) algorithm.
The name stems from the fact that the algorithm is an enhanced version of 
a pair of GDO algorithms proposed in \cite{GDO}. 
Referred to as LongestHCS and ShortestHCS, our work of \cite{GDO} shows 
that the LongestHCS algorithm outperforms ShortestHCS algorithm in most scenarios.
The main improvement of EGDO algorithm over GDO (LongestHCS) algorithm 
is its significantly improved time complexity. As described in Subsection \ref{egdoAlg}, 
the latter is achieved by locally modifying the existing MST 
in iterative steps as opposed to forming a new MST after each iteration as it was the case of GDO algorithms. 
In the rest of this paper, we refer to the LongestHCS algorithm as the GDO algorithm. 

\subsection{EGDO Algorithm}
\label{egdoAlg}
Given a number of clusters distributed in the plane, the first step is to set up the HCS origin and axes. We set the origin of the 
HCS at the center of geometry of these clusters. Then, we set up $x$- and $y$-axis for HCS at the origin, i.e., $O$. After setting up 
the HCS, the coordinates of all clusters in HCS are specified. Utilizing the HCS and the associated algebra developed in 
Section \ref{hcsSec}, the 5-step iterative algorithm \ref{alg1} shown below 
leads us to establishing network graph connectivity 
using the minimum number of intermediate ANs. 
In what follows, the individual steps of algorithm \ref{alg1} are described in detail. 
\begin{algorithm}
	\caption{EGDO Algorithm}
	\label{alg1}
	\begin{algorithmic}[5]
		\STATE \textbf{Input:} \text{Location of pre-deployed clusters}
		\STATE \textbf{Output:} \text{Coordinates of intermediate ANs} 
		\STATE \text{   }
		\STATE \text{Step 1: Establish HCS, calculate MST, find terminals} 
		\WHILE{(Graph $\mathcal{G}$ not fully connected)}
		\STATE \text{Step 2: Identify clusters $P$ and $Q$ to be connected}
		\STATE \text{Step 3: Connect $P$ and $Q$ using minimum no of ANs} 
		\STATE \text{Step 4: Modify MST using ANs placed in Step 3} 
		\STATE \text{Step 5: Break if graph is fully connected}
		\ENDWHILE
	\end{algorithmic}
\end{algorithm}
\subsection*{Step 1: Calculate MST and Find Terminals}
In this Step, we first initialize the iteration counter $k$ to $1$. 
Then, we calculate the initial distance weighted MST formed by 
a given set of clusters. The distance between a pair of clusters is calculated 
based on the distance measure definition in Section \ref{hcsSec}. The MST is
calculated using Kruskal algorithm \cite{kruskal}. 
The calculated MST is presented by an $N \times 2$ matrix $\mathbf{M}^{(k)}$ 
in which $k$ represents the iteration number, each row identifies 
the two vertices of an edge, and $N$ is the number of edges. 
Matrix $\mathbf{M}^{(k)}$ includes the edges in an 
increasing order of edge length. Given $\mathbf{M}^{(k)}$, 
we find all terminal nodes. A terminal node is an AN in 
the network connected to only one AN. 
From these terminals, we establish a set of nodes of interest for use in the next step. 
Once we place a new intermediate AN, we compare the distances between this new AN and 
the nodes in which we are interested. 
We refer to the nodes in which we are interested 
as potentially adjustable nodes (PANs). 
If the distance between the new AN and a PAN is smaller than the edge 
length connected to that PAN, we change 
the route by deleting the edge connected to the PAN and connecting the PAN with the new AN.

\subsection*{Step 2: Identify the Pair of Connecting Clusters}
Since the rows of distance matrix are in increasing order, the last row of $\mathbf{M}^{(k)}$  
identifies the edge to be connected in the next step. Let us assume the elements of 
the last row are clusters $P$ and $Q$. 

One might be curious as to why we select the pair of nodes that have the longest distance in between. The answer lies in the fact that the longest distance pair yielded our best experimental from among the variants tested. Other variants include always selecting the shortest distance pair, alternating 
between shortest and longest distance pairs, and several types of clustering strategies discussed 
in \cite{Clustering}. 

\subsection*{Step 3: Connect the Pair of Selected Clusters}
This step attempts to achieve two goals. 
The primary goal is to connect the selected pair of clusters $P$ and $Q$ 
(representing either gateway ANs or intermediate ANs) with the least number of intermediate ANs. 
The secondary goal is to deploy those intermediate ANs, which we denote as $\Upsilon$,
in a way to bring the remaining isolated clusters closer thereby helping their future connectivity. We refer to the set of intermediate ANs in iteration $k$ as $A^{(k)}$.
In order to achieve these goals, we utilize the following iterative process. 

\emph{Case 1}: When using a single AN 
suffices to connect $P$ and $Q$, $A^{(k)}=\{\Upsilon^{(k)}\}$ 
and $\Upsilon^{(k)}$ is placed  
between $P$ and $Q$. The exact position of $\Upsilon^{(k)}$ 
is calculated by solving the optimization problem given in this section. 

In order to identify the coordinates $(x,y)$ of the new node $\Upsilon^{(k)}$ 
placed in the $k$-th iteration of Case 1, we introduce a pair of conditions.
\begin{enumerate}
\item Maintain the connectivity of all three ANs, namely, $P$, $Q$, 
and the new node $\Upsilon^{(k)}$.
\item Identify the position of node $\Upsilon^{(k)}$ by maximizing the probability of connecting 
the newly aggregated cluster to other pre-deployed clusters and minimizing the overlap 
area between $\Upsilon^{(k)}$ and the other two nodes.
\end{enumerate}

In short, we want to connect the pair of selected clusters by placing $\Upsilon^{(k)}$, 
as well as expect to facilitate the connectivity of remaining clusters by intelligently placing $\Upsilon^{(k)}$ when possible. Accordingly, we formulate the following distance maximization problem graphically depicted in Fig. \ref{optim}.

\begin{eqnarray}
 \underset{x,y}{\max} & |(P,\Upsilon^{(k)})| + |(Q,\Upsilon^{(k)})|  & \label{ob1} \\
\text{S.T.} &  (P,R) \cdot (R,\Upsilon^{(k)})=0& \label{const1} \\
&  (Q,S) \cdot (S,\Upsilon^{(k)})=0 &   \label{const2} \\ 
& |({P,R})|\leq \lambda  \quad \; \;& \label{const3} \\
& |({Q,S})| \leq \lambda \quad \; \;& \label{const4} \\
& |({R,\Upsilon^{(k)}})| \leq  8n+4  & \label{const5} \\
& |({S,\Upsilon^{(k)}})| \leq  8n+4 & \label{const6}
\end{eqnarray}
In this problem, $\lambda=12n+7$ and $(P,R), (Q,S) \in \Psi$.
\begin{equation}
\Psi=\{ (\lambda,0), (-\lambda,0), (0, \lambda), (0, -\lambda), (\lambda, -\lambda),(-\lambda, \lambda)\}
\end{equation} 
By definition of $\Psi$, $\Upsilon^{(k)}$ is connected to $P$ (or $Q$) 
and has the least overlap area with $P$ (or $Q$) as long as $\Upsilon^{(k)}$  assumes its value  
from the set of feasible positions 
determined by the inner product constraint (\ref{const1}) (or (\ref{const2})). 
In other words, the distance between $\Upsilon^{(k)}$ 
and $P$ (or $Q$) is maximized along the direction of vector $(P,R)$ 
(or $(Q,S)$).
\begin{figure}[!ht]
\centering
	\includegraphics[width=0.4\textwidth]{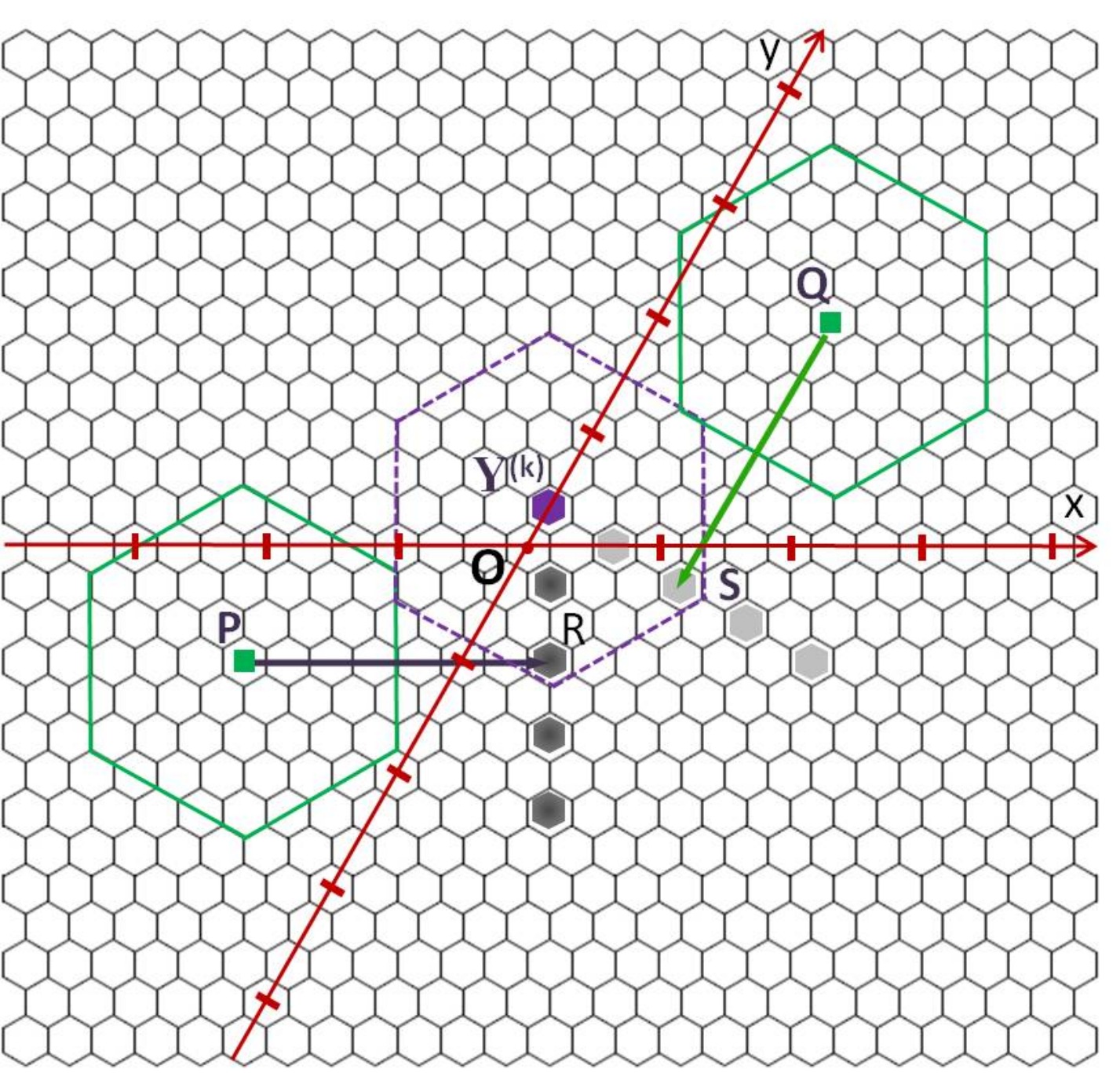}
	\caption{A graphical representation of the feasible points of $\Upsilon^{(k)}$ in Case 1.}
	\label{optim}
\end{figure}

In what follows, we explain the meaning of these constraints. 
As shown in Fig. \ref{optim}, the cells in solid dark and light grey color represent all feasible 
positions of $\Upsilon^{(k)}$ assuring connectivity to nodes $P$ and $Q$. These two inner product constraints maintain $\Upsilon^{(k)}$ slides along 
the line perpendicular to $(P,R)$ at point $R$ and $(Q,S)$ at point $S$ as shown by gray cells in Fig. \ref{optim}, while $\vert(P,R)\vert$ identifies the furthest position $\Upsilon^{(k)}$ can reach along the direction(s) of $(P,R)$ while staying connected to $P$. 
In short, the inner product constraint identifies the track of movement for $\Upsilon^{(k)}$ and $|(R,\Upsilon^{(k)})|$ controls the 
range on the track. Although $(P,R)$ has six possible directions, we do not need to inspect them all. Based on the relative position of $Q$ with respect to $P$, only one facade of each node needs to be considered. The same explanation also applies to $Q$. 
To solve the optimization problem, we first ignore the inequality constraints and calculate the position of $\Upsilon^{(k)}$ by the two equality constraints, we verify the connectivity of the two gateway nodes. 
With different pairs of $(P,R)$ and $(Q,S)$ 
selected in set $\Psi$, there will be two solutions associated with the coordinates of $\Upsilon^{(k)}$ satisfying inequality constraints. 
The one closer to the origin is chosen in order to improve the probability of connecting to other clusters.
Now, let us assume we solve Case 1 of Step 3 leading to specifying  
the coordinates of intermediate node $\Upsilon^{(k)}$ as $(x, y)$, $P$ as $(p_1, p_2)$, 
and $Q$ as $(q_1, q_2)$. Now select $(P,R)$ to be $(\lambda, 0)$ 
and $(Q,S)$ to be $(0, -\lambda)$. Then, solving the equality constraint (\ref{const1})  yields
\begin{equation}
	x + \dfrac{1}{2}y \; = \; p_1 + \lambda + \dfrac{1}{2}p_2 %
\label{optmp}
\end{equation}
Similarly, solving the equality constraint (\ref{const2})  yields
\begin{equation}
	\dfrac{1}{2}x + y \; = \; \dfrac{1}{2}q_1 + q_2 - \lambda	
\label{optmp2}
\end{equation}
By solving Eq. (\ref{optmp}) and Eq.  (\ref{optmp2}), the coordinates of $\Upsilon^{(k)}$ are identified which  are required to be 
integers within the HCS.
More importantly, the position of $\Upsilon^{(k)}$ must satisfy the inequality constraints (\ref{const3}), (\ref{const4}), 
(\ref{const5}), and (\ref{const6}). 
In this case, all possible combinations of $(P,R)$ and $(Q,S)$ in set $\Psi$ are to be inspected. 
Note that $(P,R)$ and $(Q,S)$ cannot be parallel to each other, otherwise Eq. (\ref{optmp}) and  Eq. (\ref{optmp2})  have no joint solution. 
This rule only applies to Case 1. \\

\emph{Case 2}: When using only one AN cannot establish connectivity between $P$ and $Q$, 
the following iterative process is initiated. 
In this case, we assume $A^{(k)} = \{\Upsilon^{(k)}_1, \cdots, \Upsilon^{(k)}_{\nu^{(k)} }\}$, i.e., 
$\nu^{(k)}$ ANs are required to connect $P$ and $Q$ where $\nu^{(k)} \geq 2$. 
\begin{enumerate}
\item Place an AN next to each of the two clusters $P$ and $Q$. 
Referred to as $\Upsilon^{(k)}_i$ and $\Upsilon^{(k)}_{i+1}$, $1 \leq i < \nu^{(k)}$,  
these two ANs are connected to  their associated clusters, 
$P$ and $Q$, respectively. 
\item Find the exact positions of $\Upsilon^{(k)}_i$ and $\Upsilon^{(k)}_{i+1}$ by minimizing the sum of three distances between $\Upsilon^{(k)}_i$ and $\Upsilon^{(k)}_{i+1}$, 
$\Upsilon^{(k)}_i$ and the origin, $\Upsilon^{(k)}_{i+1}$ and the origin.
\item Inspect connectivity between $\Upsilon^{(k)}_i$ and $\Upsilon^{(k)}_{i+1}$. If connected, terminate the process. If more ANs are needed, 
replace the two end nodes with $\Upsilon^{(k)}_i$ and $\Upsilon^{(k)}_{i+1}$ and then 
solve the problem of connecting them by going through the same process described before.
\end{enumerate}
\begin{figure}[!ht]
	\centering
	\includegraphics[width=0.45\textwidth]{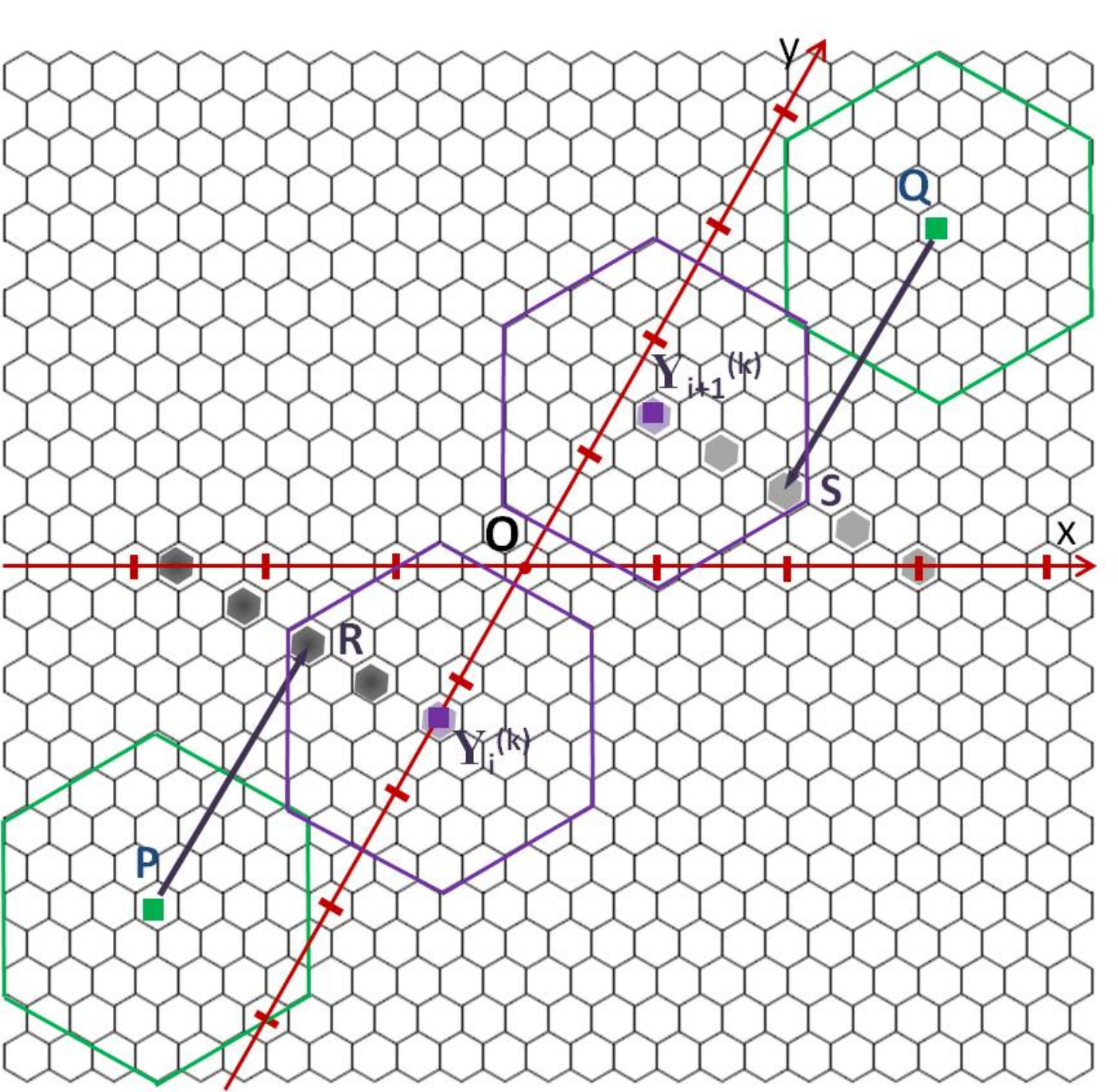}
	\caption{A graphical representation of the feasible points of $\Upsilon_i^{(k)}$ and $\Upsilon_{i+1}^{(k)}$ in Case 2.
	}
	\label{optim2}
\end{figure}
We assume the AN pair $P$ and $Q$ 
are being connected using ANs $\Upsilon_i^{(k)}$ and $\Upsilon_{i+1}^{(k)}$ with coordinates  
$(x_i,y_i)$ and $(x_{i+1},y_{i+1})$, respectively.
We use two sets of constraints to find feasible positions of $\Upsilon_i^{(k)}$ and $\Upsilon_{i+1}^{(k)}$ corresponding to $P$ and $Q$, respectively. 
Then, different combinations of these feasible positions of $\Upsilon_i^{(k)}$ and $\Upsilon_{i+1}^{(k)}$ 
are inspected in order to find the one minimizing the distance between $\Upsilon_i^{(k)}$ 
and $\Upsilon_{i+1}^{(k)}$, $\Upsilon_i^{(k)}$ and the origin, as well as 
$\Upsilon_{i+1}^{(k)}$ 
and the origin.
As shown in Fig. \ref{optim2}, dark gray cells on the line segment 
perpendicular to $(P,R)$ represent all feasible positions of $\Upsilon^{(k)}_i$.
Similarly, light gray cells on the line segment 
perpendicular to $(Q,S)$ represent all feasible positions of $\Upsilon^{(k)}_{i+1}$.
The distance optimization problem is accordingly described as below. 

\begin{eqnarray}
\underset{x_i,y_i,x_{i+1},y_{i+1}}{\min} & |(\Upsilon_{i}^{(k)},\Upsilon_{i+1}^{(k)})| + |(O,\Upsilon_{i}^{(k)}| + |(O,\Upsilon_{i+1}^{(k)})|  & \label{ob2} \\
\text{S.T.} & (P,R) \cdot (R,\Upsilon_i^{(k)})=0 & \label{coneq1} \\
& (Q,S) \cdot (S,\Upsilon_{i+1}^{(k)})=0 & \label{coneq2} \\
& |(P,R)|\leq \lambda\\
& |(Q,S)|\leq \lambda\\
& |(R,\Upsilon_i^{(k)})| \leq  8n + 4 & \label{conin1}\\
& |(S,\Upsilon_{i+1}^{(k)})| \leq  8n + 4  & \label{conin2}
\end{eqnarray}
The set of constraints used to find feasible positions of $\Upsilon_i^{(k)}$ are expressed by  
Eq. (\ref{coneq1}) and Eq. (\ref{conin1}), while the one used to find positions of 
$\Upsilon_{i+1}^{(k)}$ are expressed by Eq. (\ref{coneq2})  and Eq. (\ref{conin2}).
Here,  $(P,R)$ and $(Q,S)$ 
are in $\Psi$. 

However, $(P,R)$ 
only has one choice in $\Psi$ because the vector $(P,R)$ 
can only point to the facade that faces 
$Q$. 
Similarly, $(Q,S)$ only has one choice in $\Psi$ because the vector 
$(Q,S)$ can only point to the facade that faces $P$. 
After finding all feasible positions for 
$\Upsilon_i^{(k)}$ and $\Upsilon_{i+1}^{(k)}$, we have to conduct an exhaustive 
search among different combinations of the feasible positions to find the unique 
combination that minimizes the objective function. 

\subsection*{Step 4: Modify MST}
In this step, first the current PAN set is identified following the placement of 
one or more intermediate ANs in the previous step. In Case 1 of Step 3, 
this step takes place after the single AN is placed. In Case 2 of Step 3, this step is initiated 
right after $\Upsilon^{(k)}_i$ and $\Upsilon^{(k)}_{i+1}$ are placed. 
It is observed that placing one or two new ANs
can only introduce local changes to the topology of the network, i.e., 
topology changes are limited to the neighboring nodes of newly placed ANs. 

In a given MST, a $2$-connected node or terminal
may be connected to a newly placed ANs guaranteed not to create a loop. 
However, connecting  
a node with $3$ or more edges to a PAN may result in 
creating a loop. 
In order to avoid the 
possibility of creating a loop, only $2$-connected nodes and terminals in an MST are considered as PANs.

Accordingly, we propose a {\it line graph} method in order to identify these PANs. 
A line graph starts from a terminal node. This terminal is the first node. Then 
following the line, the second node is reached and so on. 
The line ends when a $3$-connected node, a terminal node, or a currently selected node is reached. 
If a node on the line does not belong to any of the categories above, then 
it is a $2$-connected node and is subsequently added to the current PAN set. 
Note that we are only interested in terminals or $2$-connected nodes 
because  
$3$-connected nodes cannot be modified or else the entire spanning tree will become disconnected. 
The nodes on the line represent the current PAN set which is one subset of all PANs.

Having identified the current PAN set, the MST can be modified accordingly. 
In order to always keep the last row of $\mathbf{M}^{(k)}$ 
as the edge to be connected, we modify the MST according to the cases of {\it Step 3} 
above.

If we are to follow Case 1 and place a single AN, we remove the last row $[P \; Q]$ of 
$\mathbf{M}^{(k)}$ 
and insert two rows $[P \; \Upsilon^{(k)}]$ and $[Q \; \Upsilon^{(k)}]$ as the top rows of $\mathbf{M}^{(k)}$. 
This changes the representation of MST from  $\mathbf{M}^{(k)}$ in this iteration to $\mathbf{M}^{(k+1)}$ 
in the next iteration.

If we are to follow Case 2 when adding a pair of ANs 
$\Upsilon^{(k)}_i$ and $\Upsilon^{(k)}_{i+1}$, we 
replace the last row of $\mathbf{M}^{(k)}$ with $[\Upsilon^{(k)}_i \; \Upsilon^{(k)}_{i+1}]$
and then insert $[\Upsilon^{(k)}_i \; P]$ and $[\Upsilon^{(k)}_{i+1} \; Q]$ 
as the top rows of 
$\mathbf{M}^{(k)}$ for iteration $(k+1)$. 
 
Second, for each newly added AN  
in set $A^{(k)}$,  
say $\Upsilon^{(k)}_i$, 
we compare the distance between $\Upsilon^{(k)}_i$ and 
the $j$-th node on a line graph, $\delta^{(k)}_{ij}$, 
with the distance between $j$-th and $(j+1)$-th node on the line. 
If $\delta^{(k)}_{ij}$ is smaller, we modify the tree by deleting the edge between $j$-th and 
$(j+1)$-th node on this line graph, and then adding the edge between $\Upsilon^{(k)}_i$ and $j$-th 
node. After this modification, $\Upsilon^{(k)}_i$ becomes a $3$-connected node, as $\Upsilon^{(k)}_i$ 
will be connected with $P$ and $Q$ or other intermediate ANs between $P$ and $Q$, 
as well as $j$-th node on this line graph. Therefore, the edges ending at $\Upsilon^{(k)}_i$ can 
never be modified. After making each modification, we stop searching for other PANs 
on the current or other line graphs.

In this process, we always compare $\delta^{(k)}_{ij}$ with the edge lengths of MST entries and insert the associated 
edges at the right place in order to preserve the ascending order of edge lengths in MST
matrix.  

\subsection*{Step 5: Check Stoppage Rule}
When the selected pair of clusters is found to have already been connected, the algorithm stops. 
Otherwise, we increment the iteration counter $k$ by $1$ and go back to \emph{Step 2}. \\ 

It is worth noting that the main difference between 
EGDO and GDO algorithm is the fact that EGDO algorithm adjusts 
the MST due to local topology changes associated 
with adding ANs as opposed to recalculating a new MST done by GDO. 
This leads to a significant reduction of average time complexity 
in EGDO algorithm as reported in Section \ref{secCompEGDO} and 
Section \ref{secRes}.

\subsection{Analysis of Complexity}
In this subsection, we analyze the computational complexity of EGDO in comparison with GDO algorithm. 
We first determine the time complexity of solving the optimization 
problem of Step 3 
as it is the common step shared by both GDO and EGDO algorithms, 
and then analyze the complexity of the recursive algorithm.

\subsubsection{Complexity of Solving the Optimization Problem}
The total number of cases 
that need to be inspected in order to solve the optimization problem 
of Step 3 of Section \ref{egdoAlg}
is equal to
$
{6 \choose 2} - 3 = 12 \nonumber
$.
As mentioned before, each $(P,R)$ and $(Q,S)$ has six possible directions but
cannot be in parallel or else there is no solution to Eq.(\ref{optmp}) and Eq.(\ref{optmp2}). 
However, this number can be reduced to $4$ 
based on the relative position of $P$ with respect to $Q$. 
We argue that solving the optimization problem in Case 1 of Step 3 takes constant time
as solving Eq.(\ref{optmp}) costs constant time. 

If we are to follow Case 2 in {\it Step 3}, we conduct an exhaustive search for combinations of 
feasible positions for $\Upsilon^{(k)}_i$ and $\Upsilon^{(k)}_{i+1}$. 
The total number of all individual feasible positions for $\Upsilon^{(k)}_i$ is 
\begin{equation}
\frac{2(R-r)}{3r} + 1 = \frac{2(\lambda-1)}{3} + 1 = 8n + 5
\end{equation}
The same number also represents the total number of all individual feasible positions for $\Upsilon^{(k)}_{i+1}$.
Thus, the total number of combinations that either GDO or EGDO algorithm need to inspect in Case 2 of {\it Step 3} 
is $(8n+5)^2$. 
This is an exhaustive search within a finite number of candidates since $n$ is determined by the communication range of an AN. 
Therefore, finding the particular combination of the pair  $(\Upsilon^{(k)}_i$, $\Upsilon^{(k)}_{i+1})$ 
that minimizes their distance takes constant time.  

Next, we note that Case 2 follows the same approach iteratively until $P$ and $Q$ are connected. 
This is because $\nu^{(k)}$, the total number of intermediate ANs required to connect $P$ and $Q$ in iteration $k$, 
is a finite number known at the beginning of each iteration and is decreasing in consequent iterations 
as the selected edge length within MST never increases. Hence, we conclude that
solving the optimization problem of Case 2 also takes constant time.
This constant is a function of ${R/r}$ as well as the number of ANs needed to connect the two selected nodes.

\subsubsection{Complexity of EGDO Algorithm}
\label{secCompEGDO}
In this subsection, we analyze the complexity of the other steps of the EGDO algorithm. 
In Step 1, the time complexity of calculating the distance weight matrix between all nodes 
is in the order of $\mathcal{O}(N_0^2)$ provided that 
there are $N_0$ pre-deployed clusters.
To calculate the minimum spanning tree takes $\mathcal{O}(E \log{N_0})$ where $E$ is the number of edges in the initial network graph. Since we need to inspect all edges in the weight matrix, $E$ is close to $N_0^2$. To find terminals, we need to inspect the degree of each 
node. 
Completing this process takes a time complexity of $\mathcal{O}(N_0)$.
Hence, the total complexity of Step 1 is in the order of 
$\mathcal{O}(N_0^2) + \mathcal{O}(N_0^2\log{N_0}) = \mathcal{O}(N_0^2\log{N_0})$. 
In Step 4 and in order to find the sets of PANs, we start from each terminal node 
and stop after reaching a certain type of node. This is a search process that usually stops 
way before going through all nodes at present. Assuming that the algorithm starts with $N_0$ pre-deployed clusters, stops after $n_t$ iterations, $\nu^{(k)}$ is the number of intermediate ANs added in iteration $k$ with $\nu^{(0)}=0$, and $\Gamma(k)= N_0 + \sum_{i = 0}^k \nu^{(i)}$ represents the total number of ANs after iteration $k$. Thus, the worst case time complexity of Step 4 is $\mathcal{O}(\Gamma(k-1))$ at $k$-th iteration. However, the average time complexity 
is much shorter as the search stops way before $\Gamma(k-1)$.
Step 2 and Step 5 takes constant time which can be ignored.

In summary, the worst case 
time complexity of EGDO algorithm is in the order of 
\begin{equation*}
\mathcal{O}(N_0^2\log{N_0}) + \sum_{k=1}^{n_t} \mathcal{O}(\Gamma(k-1)) < \mathcal{O}\left[ N_0^2\log{N_0} + n_t\Gamma(n_t)\right]
\end{equation*}
However, the average time complexity is much shorter considering the fact that the search process of Step 4 stops 
before $\Gamma(k-1)$ 
as showed by our experiments in Section \ref{secRes}.

\subsubsection{Complexity of GDO Algorithm}
In comparison, we analyze the complexity of the other steps of GDO algorithm. In Step 1
the time 
complexity of calculating the distance weight matrix between all nodes is $\mathcal{O}(N_0^2)$ provided that there are 
$N_0$ pre-deployed clusters. 
Similarly, 
the time complexity of Step 1 is in the order of 
$\mathcal{O}(N_0^2 \log{N_0})$.
Since the GDO algorithm recalculates the MST after each iteration and the number of nodes in MST is increasing, the runtime accumulates.  
With the same definition of $n_t$, $ \nu^{(k)}$, and $\Gamma(k)$,
the time complexity of GDO algorithm is 
\begin{equation*}
\sum_{k = 0}^{n_t} \mathcal{O}\left[ \Gamma(k)^2\log{\Gamma(k)}\right]  < \mathcal{O}\left[ n_t\Gamma(n_t)^2 \log{\Gamma(n_t)} \right] 
\end{equation*}
The worst case time complexity of GDO algorithm is hence in the order of 
$\mathcal{O}\left[ n_t\Gamma(n_t)^2 \log{\Gamma(n_t)} \right] $. 
Even though the bound is not tight, 
GDO algorithm has a much higher time complexity than EGDO
as presented in Section \ref{secRes}. 

\section{Experimental Results}
\label{secRes}
In this section,  
we first compare the results of our algorithm with those of exhaustive search algorithm. The latter serves as the benchmarking baseline finding the global optimal solution to the problem 
of node placement albeit with a very high time complexity. 
We show that our algorithm provides results close to the optimal solution given 
by exhaustive search within a limited area where the time complexity of exhaustive search is affordable.
Then, we compare the performance of EGDO algorithm with those of GDO and  
variants of the SMT \cite{MSTh} algorithm. 
Our experimental results cover AN cost, 
i.e., the number of intermediate ANs, runtime, robustness, and the effects of HCS.

\subsection{Comparison with Exhaustive Search Algorithm}
In this subsection, 
we compare the AN cost of EGDO algorithm with exhaustive search algorithm, 
without considerations of runtime,
in order to show our EGDO algorithm is in fact producing results close to the 
global optimal solution.

In order to examine the deviation of EGDO solution from the globally optimal solution, we run experiments in a field of $4500 \times 4500m^2$, with $r = 50m$ and $R = 350m$. The selection of parameters allows for completing exhaustive search experiments in realistic time. Fig. \ref{greedy} gives the results. The horizontal axis is the number of 
pre-deployed clusters varying in the range from $2$ to $25$ and the vertical axis is the AN cost. 
For each point on the $x$-axis, we run 50 different configurations and record the number of intermediate ANs used. Then, we fit the data to a polynomial curve. The blue and black 
curves show the AN cost of EGDO and exhaustive search algorithms, respectively. 
While the AN cost of EGDO is always higher than that of exhaustive search, 
the largest gap observed between two curves along the vertical axis is less than $10\%$.
Without being able to offer a mathematical proof, the gap falls in the 
range of $1.1$-approximation ratio.
It can also be observed that both curves start dropping beyond a certain point. This is because when the number of pre-deployed clusters grows, the network becomes denser and requires less ANs to establish connectivity.  This aspect will be further investigated in the following subsections. 

On the aspect of runtime, the completion time of EGDO algorithm is in the range of $1$ to 
$2$ seconds in our simulation setting. However, the exhaustive search algorithm 
takes from several minutes to over ten hours to complete within the same simulation settings.
\begin{figure}[!ht]
\centering
\includegraphics[width = 0.45\textwidth]{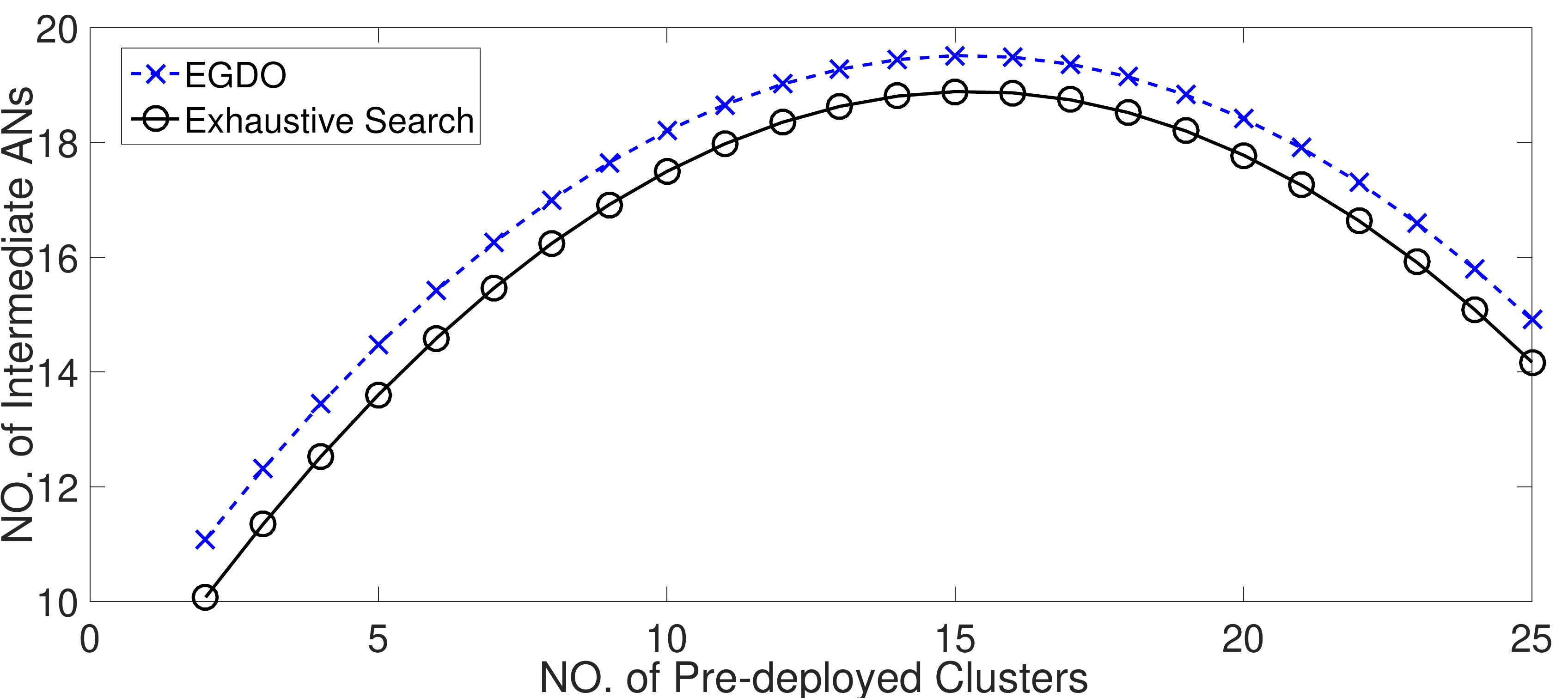}
\caption{An AN cost comparison of EGDO and exhaustive search algorithms.}
\label{greedy}
\end{figure}

\subsection{Performance Comparison of SMT and EGDO Algorithm}
In this subsection, we compare the performance of SMT, GDO, and EGDO algorithms 
measured by the minimum AN cost and 
runtime.
When comparing the two classes of algorithms, we also consider the fact that 
GDO and EGDO algorithms dynamically
update minimum spanning trees while the original SMT algorithm forms the 
minimum spanning tree once statically. 
Therefore, we modify the SMT algorithm to become a dynamic algorithm in which the minimum spanning tree is recalculated after connecting every edge. We refer to the original static SMT algorithm as StaSMT and the revised 
dynamic SMT algorithm as DynSMT. Because SMT 
algorithms model the communication range of a node as a disk 
while GDO and EGDO do so as a hexagon embedded within the disk, 
SMT algorithms cover distance with a smaller number of ANs in average. 
Yet one should notice that this is not a fair comparison as a circle always covers a longer distance than the hexagon embedded in it. 
Additionally, it is important to note that DynSMT and GDO algorithms recalculate 
the entire spanning tree each time after a pair of clusters are connected. Hence, the time complexity 
of these algorithms is higher than those of StaSMT and EGDO algorithms.
\begin{figure}[!ht]
	\centering
	\includegraphics[width = 0.5\textwidth]{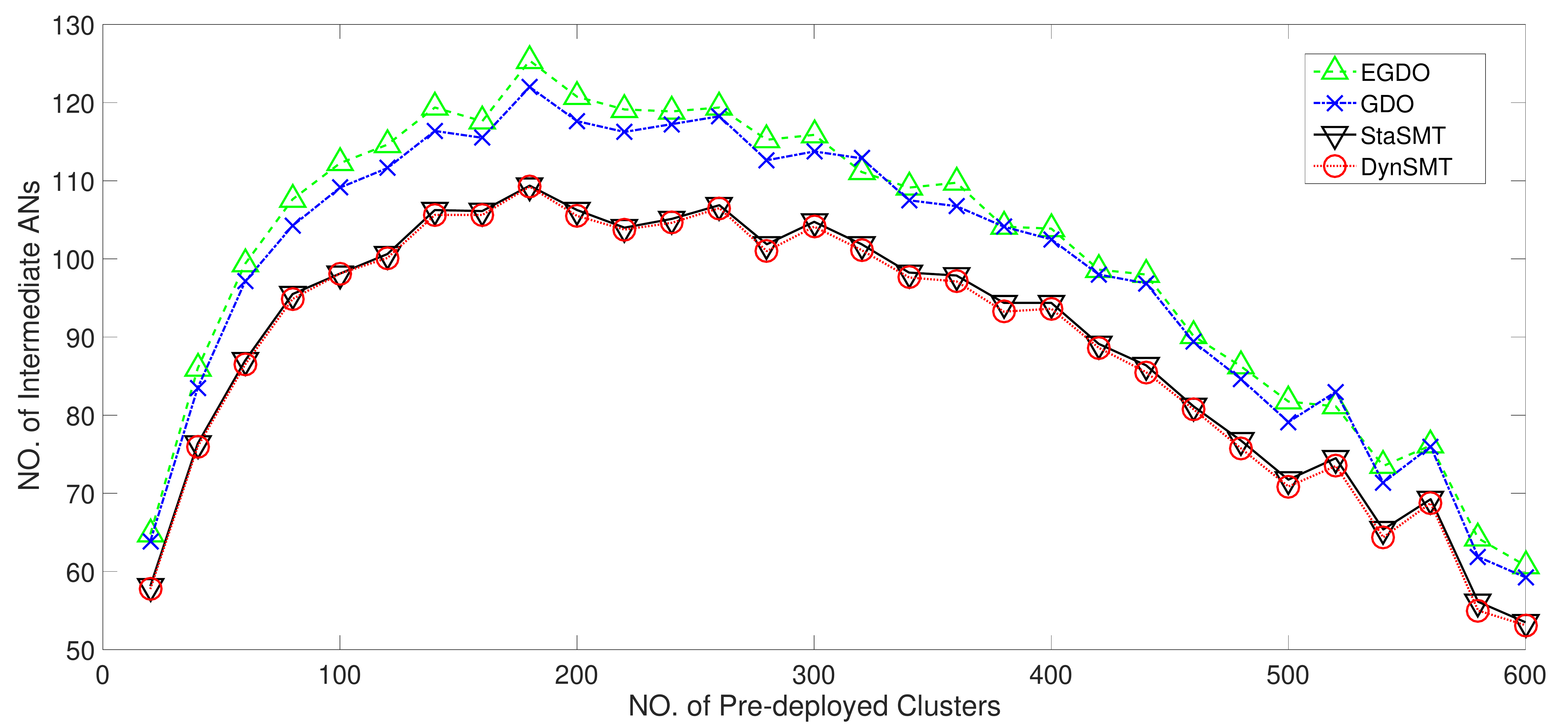}
	\caption{An AN cost comparison among StaSMT, DynSMT, GDO, and EGDO algorithms.}
	\label{4Alg}
\end{figure}

The experiments are conducted in a $200km \times 200km$ field with $r = 50m$ and 
$R = 4550m$. 
Fig. \ref{4Alg} provides an AN cost comparison among the StaSMT, DynSMT, GDO, and EGDO algorithms 
for a fixed field size. The 
results show that GDO algorithm uses an average  
of $10\%$ more AN resources than StaSMT. 
Further, the EGDO algorithm sometimes consumes a slightly larger number of AN resources 
than GDO algorithm, because in doing local modification it might miss some larger scale 
variations in network graph topology caused by a newly added AN. However,  
comprehensive experimental results have shown that these 
differences are negligible. 
Interestingly, it is also observed that there is no significant difference between the performance 
of the two variants of the SMT algorithm. This is alluded to the fact that unlike the GDO 
algorithm, the recalculated minimum spanning tree in DynSMT is not much different from 
the previously calculated minimum spanning tree obtained by StaSMT algorithm.
The results of all four algorithms show an initial rise followed by a drop alongside some variations. 
The rise is related to the fact that an increase in the number of pre-deployed clusters $N$
in a sparse network requires utilizing more intermediate ANs. As the $N$ grows even larger 
within a fixed field size, the sparse network evolves to a dense network covering most of 
the field with AN gateways thereby reducing the number of intermediate ANs. All four algorithm 
tend to use the same number of intermediate ANs as the value of $N$ grows to $600$ in this setting.
\begin{figure}[!ht]
	\centering
	\includegraphics[width = 0.5\textwidth]{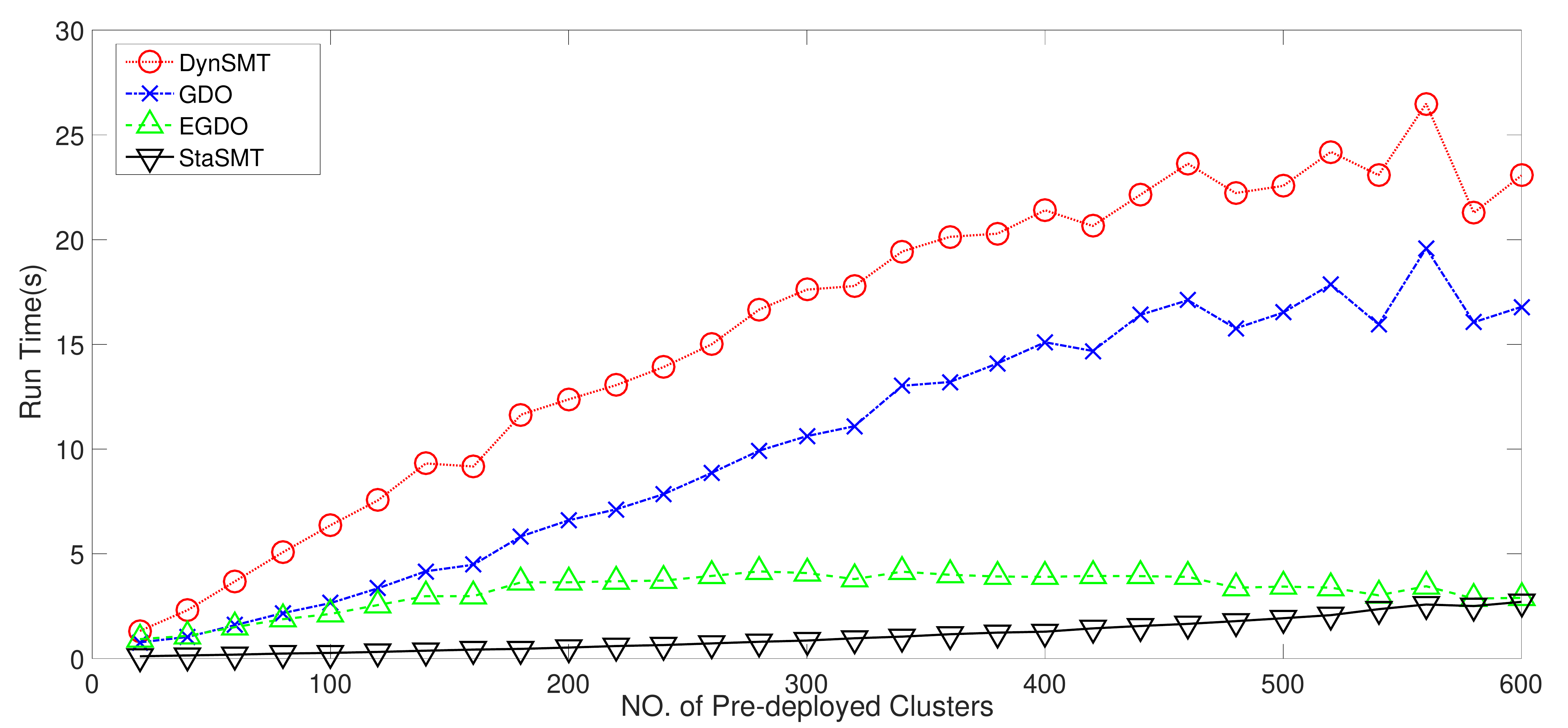}
	\caption{A runtime comparison of StaSMT, DynSMT, GDO, and EGDO algorithms.}
	\label{4Algtime}
\end{figure}

Fig. \ref{4Algtime} includes a comparison of runtimes among the StaSMT, DynSMT, GDO, and EGDO algorithms 
for the same fixed field size. 
It is observed that StaSMT has the lowest runtime because it only forms the minimum spanning tree once.
DynSMT has a much longer runtime than the other three algorithms in general.
Among three dynamic algorithms GDO, EGDO, and DynSMT, EGDO has the shortest runtime by far.
While the runtime is generally higher than that of StaSMT, it gets closer to that of StaSMT for values of $N$ greater than $500$. This behavior is related to the fact that the cost of calculating the minimum spanning tree increases as $N$ grows and also a smaller number of intermediate nodes are needed.

Considering the fact that the AN cost performance of DynSMT is slightly better than that of StaSMT but its runtime is significantly longer, we conclude that the advantage of DynSMT does not justify its increased time complexity. 
Therefore, we mainly compare the performance of StaSMT and EGDO algorithms in the rest of our experiments,
considering comparable performance of GDO and EGDO but much better time complexity of EGDO. 
We note that EGDO algorithm uses an additional $10\%$ AN resources in average 
due to the use of a hexagon instead of a circle to represent the communication range 
of a node and also has a slightly longer runtime compared to StaSMT algorithm. 
However, it offers much better robustness characteristics as reported in the next subsection.

\subsection{Partial and Global Robustness Tests}
In this subsection, we evaluate the robustness of network connectivity algorithms 
by applying perturbations to the position of nodes. 
In each experiment, we first establish global network connectivity applying EGDO and StaSMT algorithms. 
Once connectivity is established, we introduce random perturbations to the position 
of pre-deployed clusters. This scenario is referred to as partial perturbation as it does not 
perturb the position of intermediate ANs added for establishing connectivity. 
We also conduct additional robustness experiments
in which all existing ANs after node placement are perturbed. 
We refer to such experiments as global perturbation experiments. 
A perturbation constitutes a random directional displacement of the AN from its original 
position 
by a fixed distance $4r$. The fixed value of perturbation displacement $4r$, albeit in random direction, 
represents the experimental finding within the topology of our experiments introducing the 
most pronounced impact on network connectivity without completely partitioning the network.
In each experiment and after applying perturbation, we test global connectivity.

We conduct our experiments in different field sizes but report sample results for a $200 km \times 200 km$ field. 
The set of pre-deployed clusters are distributed randomly following a uniform Poisson point process in the field of experiment. 
We set parameters $r$ and $R$ at $50m$ and $4550m$, respectively. 
The number of clusters varies from $20$ to $160$ by a step size of $20$. 
Before reporting our results, we define a measure to quantify robustness. 
Equation (\ref{RoFa}) gives the definition of the measure referred to as robustness factor (RF). 
The RF measure not only takes into consideration the probability of staying connected after 
perturbation, but also the number of intermediate ANs used to establish connectivity. 
\begin{equation}
RF = (Pr_{EGDO} - Pr_{SMT}) \times	\dfrac{\eta_{SMT}}{\eta_{EGDO}}
\label{RoFa}
\end{equation}  
In Equation (\ref{RoFa}), $Pr_{EGDO}$ and $Pr_{SMT}$ represent the probabilities 
of remaining connected after perturbation is applied to 
the cases of EGDO and SMT algorithms, respectively.
Accordingly, the calculation of $Pr_{EGDO}$ in perturbation tests is described below. 
In each experiment, the global connectivity count is increased by one if the network remains connected after 
applying perturbation.
The value of $Pr_{EGDO}$ is identified 
by dividing the global connectivity count to the total number of experiments, which is 500 here. 
Similarly, $Pr_{SMT}$ is identified. The numbers $\eta_{EGDO}$ and $\eta_{SMT}$ represent 
the number of intermediate ANs used to establish global connectivity in EGDO and SMT algorithms. 

Because EGDO algorithm uses hexagons instead of circles, it generally covers a given distance along 
a line with a larger number of ANs than SMT. However, placing nodes towards the center of geometry 
within HCS offsets some of the impact. Generally speaking,  
the EGDO algorithm is observed to use a larger number 
of intermediate ANs than StaSMT. In return, it offers a higher level of robustness.
\begin{figure}
	\centering
	\includegraphics[width=0.47\textwidth]{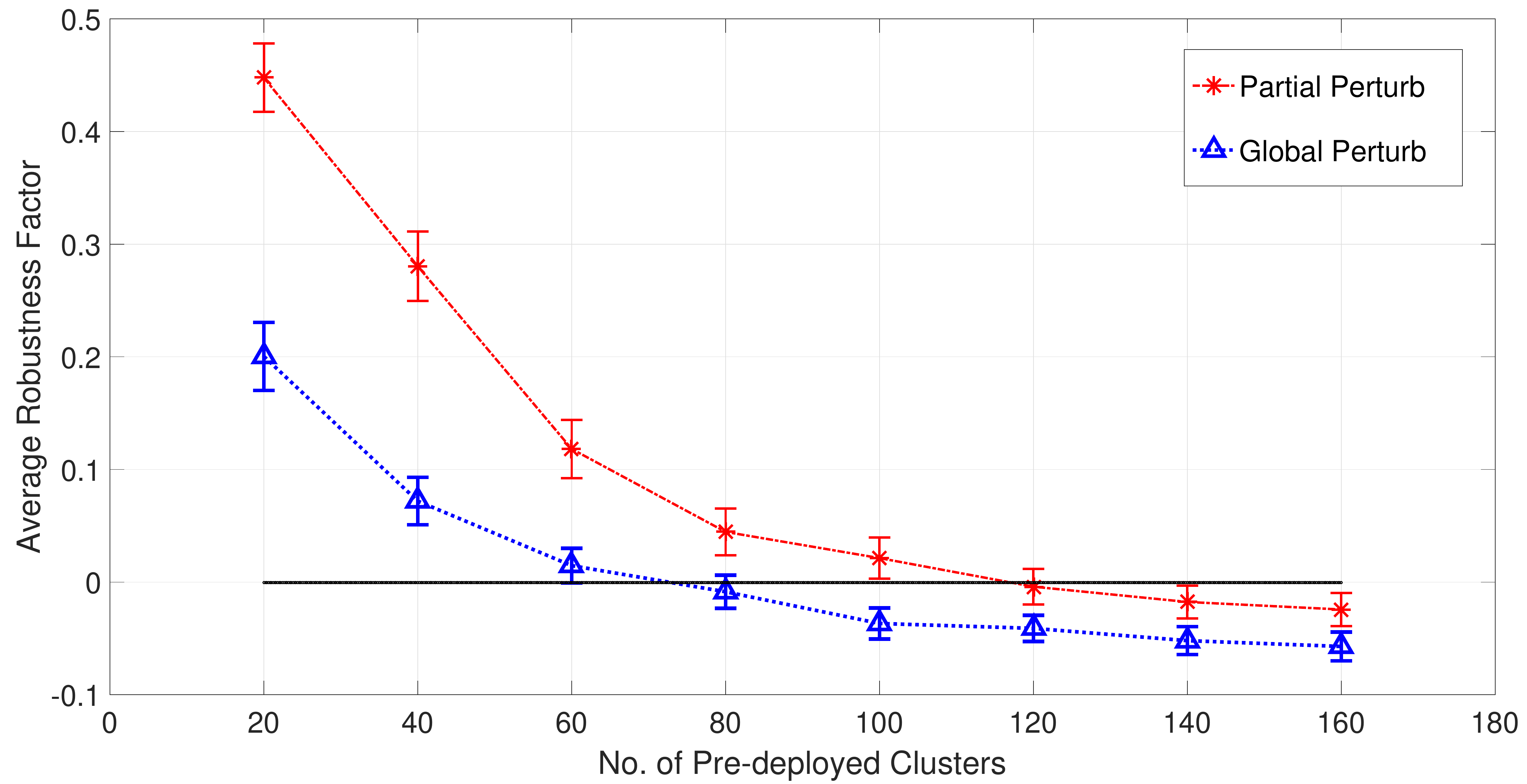}
	\caption{A drawing of average robustness factor as a function of the number of pre-deployed clusters in perturbation tests.}
	\label{GlbPtb}
\end{figure}

Experimental results of partial perturbation 
within $95\%$ confidence intervals are shown by red line in Fig. \ref{GlbPtb}. 
The horizontal axis shows the number of pre-deployed  disconnected clusters before 
we apply any node placement algorithm. The vertical axis is the 
value of RF averaged over 500 different scenarios at each given number of pre-deployed clusters. We notice that the 
value of RF is in the range $[-1, 1]$ as two probability measures are within $[0, 1]$ and the EGDO algorithms is expected 
to use a larger number of ANs than the SMT algorithm. A positive value of RF closer to 1 means that EGDO algorithm 
achieved much better robustness characteristics while using a relatively small number of ANs. 
An inspection of the reported results of Fig. \ref{GlbPtb} reveals that the EGDO algorithm shows a significant 
performance advantage in sparse networks. However and as the number of pre-deployed clusters increases, 
there is a threshold of cluster density beyond which EGDO algorithm will lose its 
advantage over SMT algorithm. More information about the threshold will be given in the next subsection.

Besides partial perturbation tests, we also conduct global perturbation experiments. In these tests, we perturb the 
positions of pre-deployed AN gateway nodes as well as intermediate AN nodes. All ANs within the connected network 
graph are displaced along a random direction by an amplitude of $4r$. The value of RF is calculated in the same way 
as explained before. 
The test results within $95\%$ confidence intervals  
are shown in Fig. \ref{GlbPtb} by the blue curve. Compared to partial perturbation test results, the RF 
values in global perturbation tests show a lower starting point and a faster drop rate as the density of clusters grows higher. 
The results show that the difference in perturbation robustness is very significant in some scenarios. Specifically, 
it is observed that the value of $Pr_{EGDO}$ is one to two orders of magnitude larger than the value of $Pr_{SMT}$ 
in some instances. 

\subsection{Inspection of Cluster Density Threshold Value}
\begin{figure}[!ht]
	\centering
	\begin{subfigure}{0.48\textwidth}
		\includegraphics[width=\textwidth]{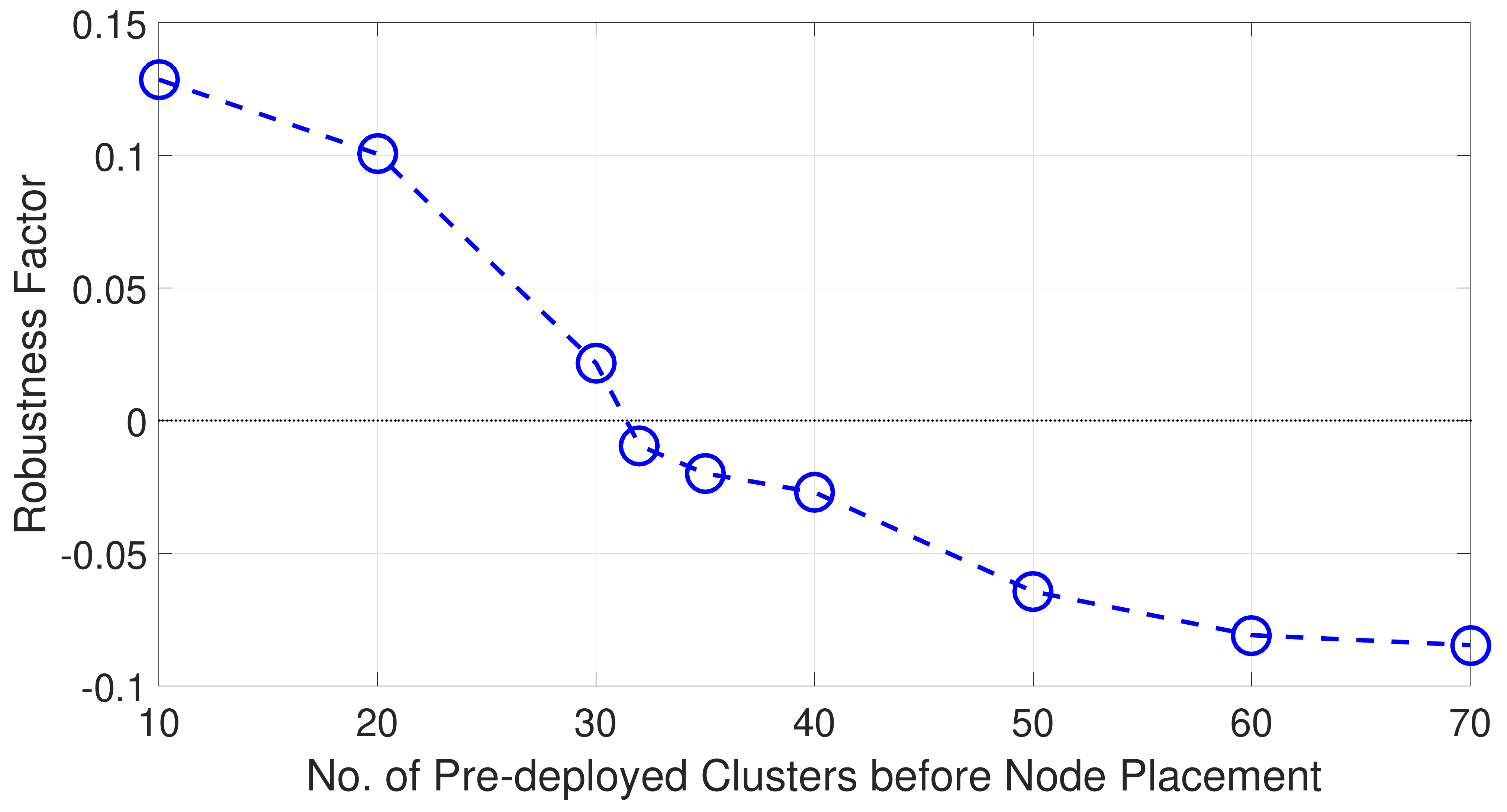}
		\caption{$100km \times 100km$}
		\label{F11th}
	\end{subfigure}
	\begin{subfigure}{0.5\textwidth}
		\includegraphics[width=\textwidth]{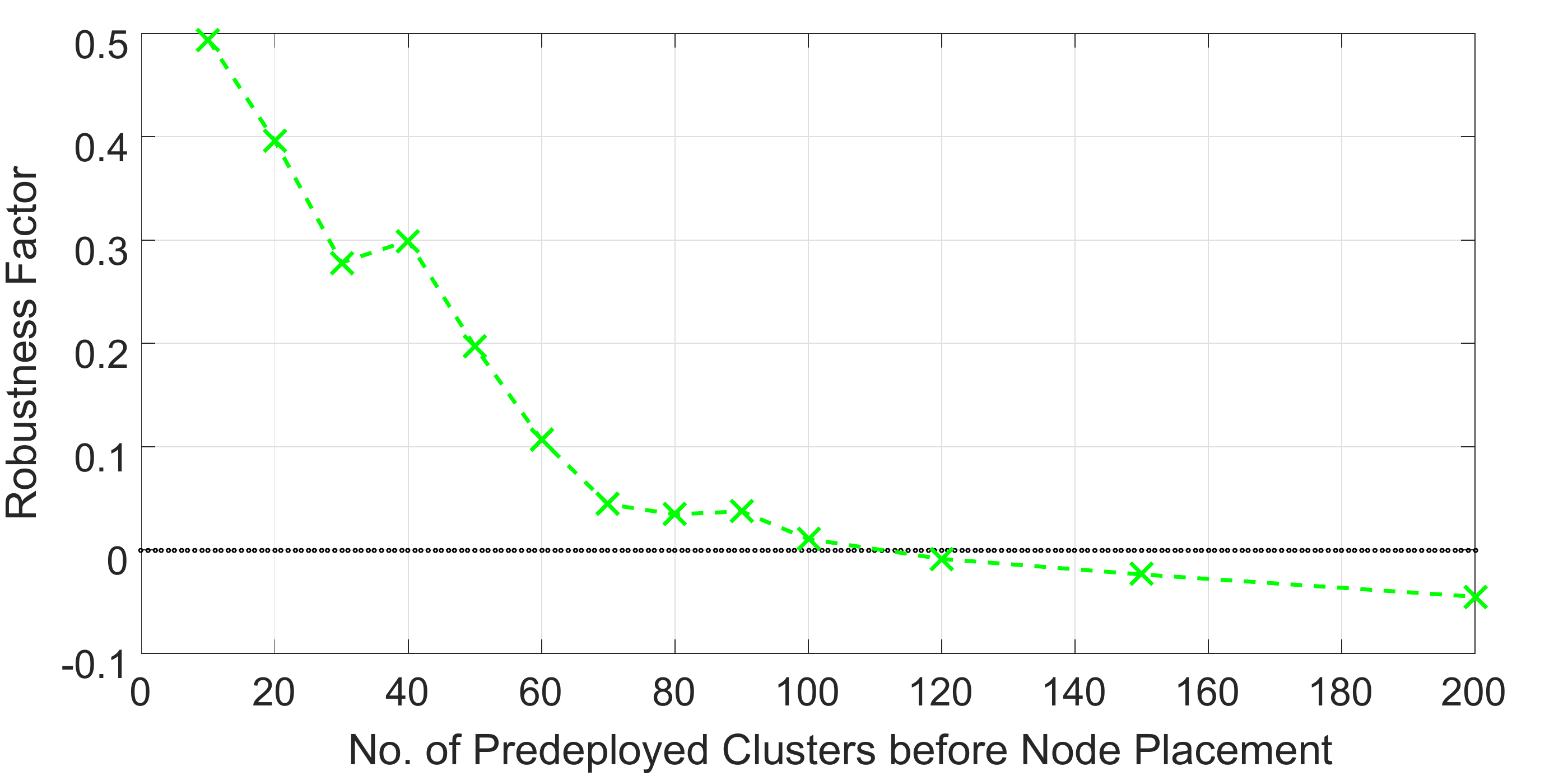}
		\caption{$200km \times 200km$}
		\label{F22th}
	\end{subfigure}
	\begin{subfigure}{0.48\textwidth}
		\includegraphics[width=\textwidth]{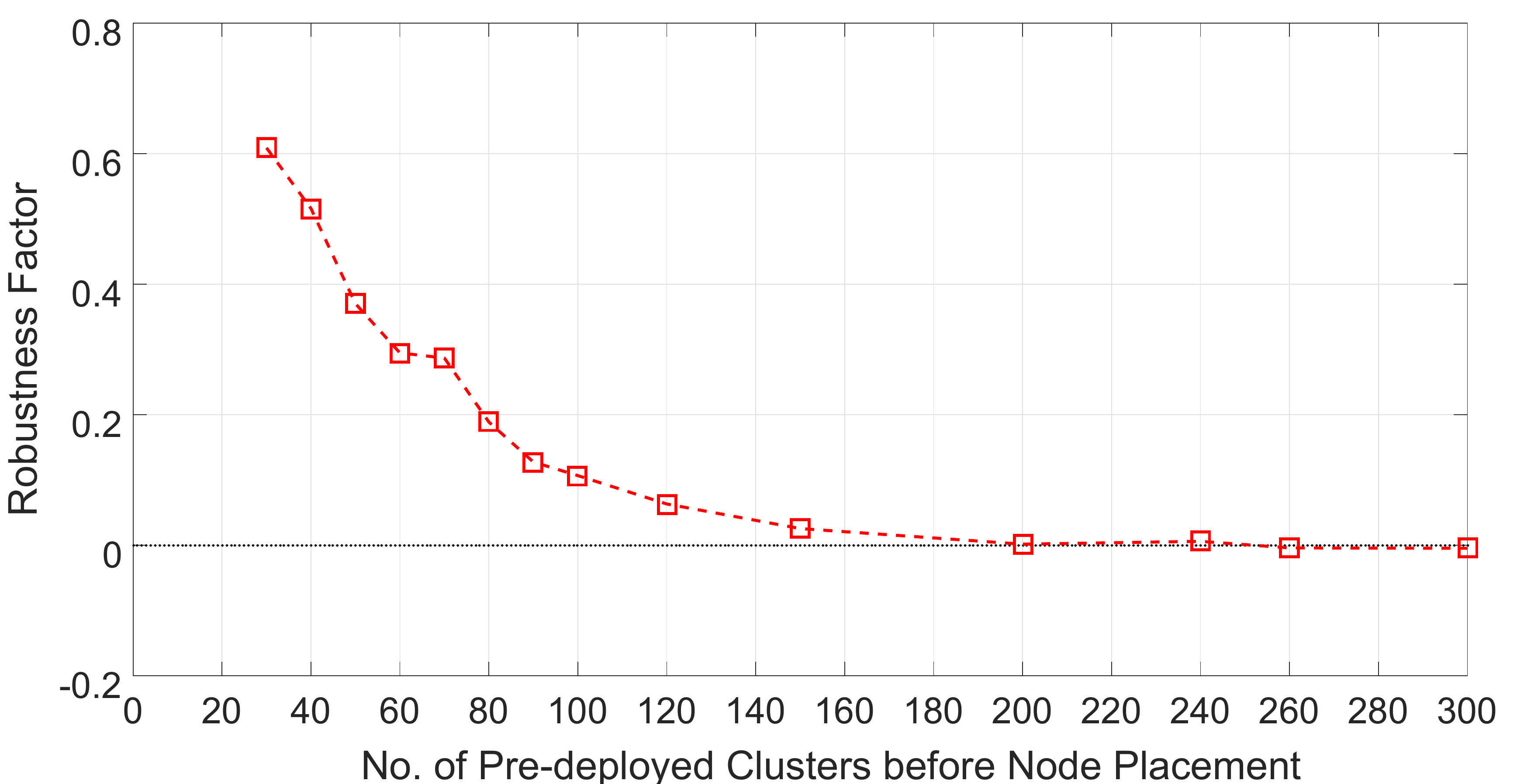}
		\caption{$300km \times 300km$}
		\label{F33th}
	\end{subfigure}
	\caption{The identification of threshold $\tau$ for different field sizes.}
\end{figure}

As described in the previous subsection, we observe a threshold of AN density beyond which the network can no longer 
be regarded as sparse. The threshold to which we refer as $\tau$ denotes a cluster density value 
passed which the EGDO algorithm offers no advantage compared to the SMT algorithm. 
In this subsection, we raise a hypothesis that the value of threshold $\tau$ is related to the density of ANs, 
namely, the field area divided by the total area of AN coverage. 
We note that both SMT and EGDO algorithms seek to minimize the AN cost. Yet, the SMT algorithm attempts at reducing the total 
distance covered by ANs while the EGDO algorithm tries to reduce the AN overlap areas. 
In essence, minimizing the area of overlap is no longer meaningful when the AN density goes beyond a certain value. 
As cluster density grows, the average overlap area increases. Thus, the robustness of 
SMT algorithm will inherently improve and EGDO algorithm no longer offers any robustness advantage.
To numerically validate this hypothesis, we conduct experiments on three different field sizes, of $100km \times 100km$, $200km \times 200km$, and $300km \times 300km$. We apply the partial perturbation test to each field and vary the 
number of pre-deployed clusters. The threshold value for each field size is identified as where the plots of RF versus AN 
cross the horizontal axis. Perturbation experiments are repeated 100 times in each scenario and for every number of 
clusters. Further, we test 100 different scenarios and report the average results. 
In Fig. \ref{F11th}, Fig. \ref{F22th}, and Fig. \ref{F33th}, the RF curves approximately cross 
the $x$-axis at values of $32$, $120$, and $260$.

Table \ref{tb1} records average intermediate AN cost for each given number of pre-deployed clusters in the test of the 
$100km \times 100km$ field. Table \ref{tb2} and Table \ref{tb3} show the AN cost in the tests of  
$200km \times  200km$ and $300km \times 300km$ field sizes, respectively. 
As described above, the threshold is defined as
\begin{equation}
\tau \propto \dfrac{\text{Field Area}}{\text{AN area} \times \text{No. of ANs}}
\end{equation}
The threshold values $\tau_1$, $\tau_2$, and $\tau_3$ are calculated below for $100km \times 100km$, 
$200km \times 200km$, and $300km \times 300km$ field size scenarios where $c$ absorbs all constants.
\begin{equation*}
\tau_{1} \approx c\times \frac{10^{10}}{4550^2 \times	(32 + 27.1)} = 0.017 \times  \frac{10^{10}}{4550^2}c
\end{equation*}
\begin{equation*}
\tau_{2} \approx c\times \frac{4\times 10^{10}}{4550^2 \times	(120 + 102.7)} = 0.018 \times \frac{10^{10}}{4550^2}c
\end{equation*}
\begin{equation*}
\tau_{3} \approx c\times \frac{9\times 10^{10}}{4550^2 \times	(260 + 227.2)} = 0.018 \times \frac{10^{10}}{4550^2}c
\end{equation*}
From the calculations, 
the values of $\tau_{1}$, $\tau_{2}$, and $\tau_{3}$ are all around to $0.018c$. 
While not reported here, we have observed similar patterns with different values of $r$, $R$, and field sizes.
The results numerically support our hypothesis that the value of threshold $\tau_i$ 
is related to the ratio of the field area and the total area covered by ANs.

\begin{table}[!ht]
\caption{Average AN costs in different field sizes.}
\centering
\begin{subtable}{0.5\textwidth}
\centering
\begin{tabular}{|c||c|c|}
 \hline
No. of Clusters & EGDO AN cost & SMT AN Cost \\
 \hhline{|-|-|-|}
 10 & 20.2 & 18.1 \\
  \hline
 20 & 26.6 & 24.0 \\
  \hline
 30 & 29.2 & 26.1 \\
  \hline
 \textbf{32} & \textbf{29.7} & \textbf{27.1} \\
  \hline
 35 & 30.3 & 27.3 \\
  \hline
 40 & 30.9 & 28.6 \\
  \hline
 50 & 30.9 & 28.2 \\
  \hline
 60 & 30.5 & 28.0 \\
  \hline
 70 & 30.2 & 28.4 \\
  \hline
 80 & 29.5 & 27.2 \\
  \hline
\end{tabular}
\caption{Measures of AN cost in $100km \times 100km$ field test.} 
\label{tb1}   
\end{subtable}
\begin{subtable}{0.5\textwidth}
\centering
\begin{tabular}{|c||c|c|}
 \hline
No. of Clusters & EGDO AN cost & SMT AN Cost \\
 \hhline{|-|-|-|}
 10 & 46.5 & 41.6 \\
  \hline
 20 & 63.7 & 56.9 \\
  \hline
 30 & 74.8 & 67.4 \\
  \hline
 40 & 84.1 & 75.5 \\
  \hline
 50 & 91.4 & 81.9 \\
  \hline
 60 & 96.0 &  86.2 \\
  \hline
 70 & 100.4 & 90.8 \\
  \hline
 80 & 103.9 & 93.8 \\
  \hline
 90 &  105.2 & 95.2 \\
  \hline
 100 & 109.0 & 99.8 \\
  \hline
 \textbf{120} & \textbf{112.9} & \textbf{102.7} \\
  \hline
 150 & 116.3 & 106.4 \\
  \hline
 200 & 117.9 & 109.1 \\
 \hline
\end{tabular} 
\caption{Measures of AN cost in $200km \times 200km$ field test.} 
\label{tb2}
\end{subtable}
\begin{subtable}{0.5\textwidth}
\centering
\begin{tabular}{|c||c|c|}
 \hline
No. of Clusters & EGDO AN cost & SMT AN Cost \\
 \hhline{|-|-|-|}
 30 & 127.8 & 120.7 \\
  \hline
 40 & 136.5 & 121.6 \\
  \hline
 50 & 148.7 & 133.3 \\
  \hline
 60 & 161.8 & 145.3 \\
  \hline
 70 & 170.6 & 153.0 \\
  \hline
 80 & 179.0 & 159.8 \\
  \hline
 90 & 187.0 & 169.3 \\
  \hline
 100 & 193.2 & 173.5\\
  \hline
 120 & 205.4 & 184.8\\
  \hline
 150 & 218.3 & 197.0 \\
  \hline
 200 & 233.8 & 211.3 \\
 \hline
 240 & 244.3 & 222.3 \\
 \hline
 \textbf{260} & \textbf{247.3} & \textbf{227.2} \\
 \hline
 300 & 252.1 & 231.6 \\
 \hline
\end{tabular}
\caption{Measures of AN cost in a $300km \times 300km$ field test.} 
\label{tb3}
\end{subtable} 
\end{table}

\subsection{An AN Cost Comparison of SMT and EGDO in HCS}
Since EGDO algorithm utilizes hexagonal tiles instead of radial disks to model the range of advantaged nodes, one can 
raise the question as to what happens when applying SMT algorithm to a network using hexagonal tiling. In order to answer this question, we run an additional experiment. 
\begin{figure}[!ht]
	\centering
	\includegraphics[width=0.5\textwidth]{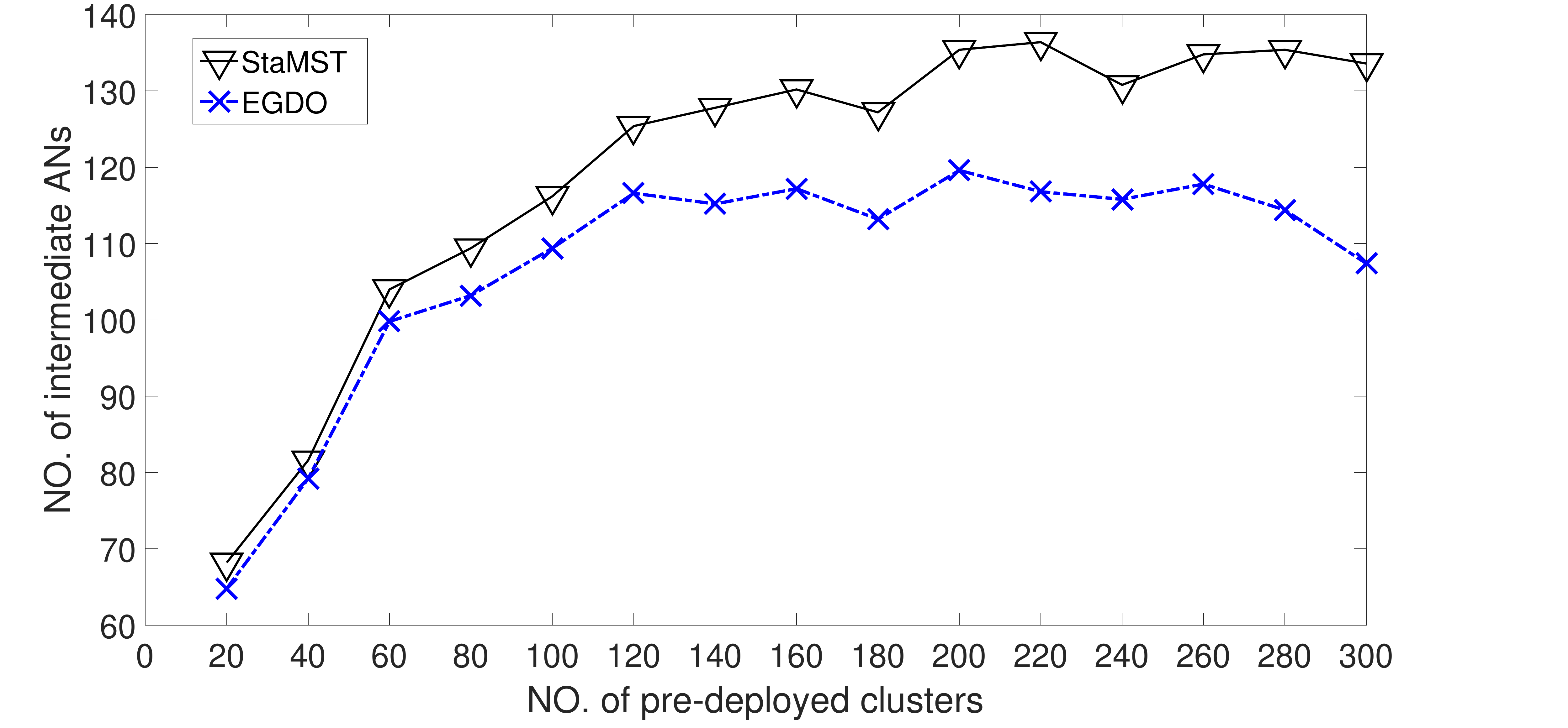}
	\caption{An AN cost comparison of SMT and EGDO algorithms in HCS.}
	\label{GDOMSTperf}
\end{figure}

Our experimental setting  is described as follows. Within an area of $120km \times 120km$, 
we randomly deploy a number of clusters ranging from $50$ to $500$ at an increasing step size of $50$. 
In this experiment, we set $r$ and $R$ at $30m$ and $2730m$, respectively. 
For each fixed number of pre-deployed clusters, we run $100$ different randomly distributed scenarios. 
Then, we average the number of ANs to report our results. 
Fig. \ref{GDOMSTperf} compares the AN cost of establishing connected graphs, 
through SMT and EGDO algorithms with the same level of built-in robustness,
as a function of the number of pre-deployed clusters. 
In this setting, the network is no longer considered sparse 
once the number of pre-deployed clusters reaches $250$.

It can be observed from the results that EGDO algorithm performs slightly better
when the number of clusters is small. As the number of pre-deployed clusters grows, the EGDO algorithm intends to 
use even a smaller number of ANs than the SMT algorithm to establish full connectivity. The number of ANs used by 
the EGDO algorithm is typically $10\%$ to $20\%$ less than those used by the SMT algorithm
for as long as the network is sparse, i.e., the number of pre-deployed clusters is less than $250$. 
Interestingly, the AN cost advantage of the EGDO algorithm becomes even more apparent
for a dense network with more than $250$ pre-deployed clusters. 
However, the advantages of EGDO over SMT in dense networks are not of
high significance because a dense network naturally offers robustness.

While not shown here, it is also important to note that representing the communication 
range of an AN with a reduced radius circle or a reduced edge square in CCS
leads to utilizing an increased number of ANs in establishing connectivity.

\section{Conclusion}
In this paper, we investigated robust connectivity in two-tiered heterogeneous network graphs through systematic placement of advantaged nodes. Our method was developed utilizing a so-called hexagonal coordinate system (HCS) in which we developed an extended algebra. We formulated and solved (within bounds) an 
NP-hard problem addressing graph connectivity. Further, we developed a class of 
near-optimal yet low complexity geometric distance optimization (GDO) algorithms 
approximating the original problem. 
Experimental results showed the effectiveness of our proposed GDO algorithms 
measured in terms of advantaged node cost and robustness of connectivity 
in sparse networks in comparison
with variants of exhaustive search and 
Steiner minimum tree (SMT) algorithms. %
Our experimental results also offered a couple of additional important insights. 
First, 
it was commonly observed that our 
proposed GDO algorithms lost their advantages in comparison with SMT algorithms 
past a density threshold value due to the higher density of nodes.
Second, below the specific sparsity threshold, 
our proposed algorithms used smaller numbers of AN nodes if 
we applied HCS representation to SMT algorithms in order to improve robustness.

\begin{IEEEbiography}%
[{\includegraphics[width=1in,height=1.25in,clip,keepaspectratio]{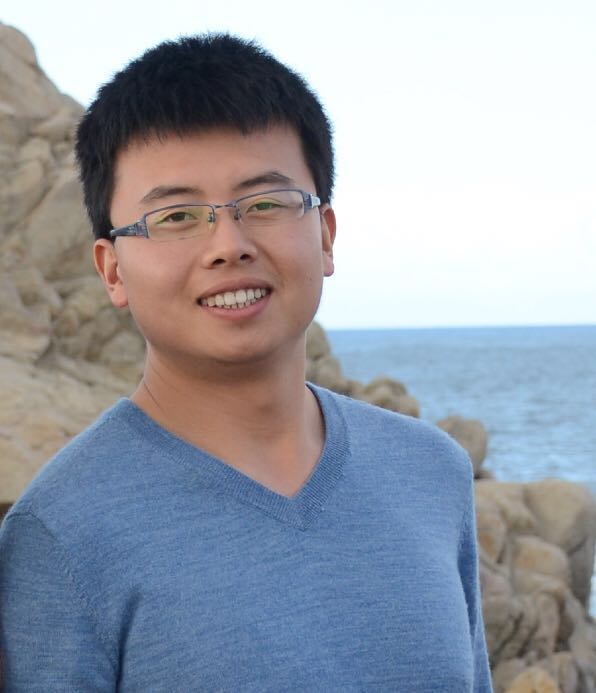}}]{Kai Ding}
received his B.S degree in Beijing University of Aeronautics and Astronautics(BUAA), Beijing, China, in 2012. He earned his M.S. degree in 2014 from the Department of Mechanical and Aerospace Engineering in UC, Irvine. Currently, he is pursuing Ph.D. degree in the same department. His research interests are in the areas of network control, wireless communications, and control systems. 
\end{IEEEbiography}
\begin{IEEEbiography}%
[{\includegraphics[width=1in,height=1.25in,clip,keepaspectratio]{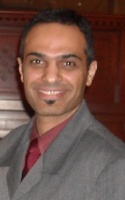}}]{Homayoun Yousefi'zadeh} (SM) 
received E.E.E and Ph.D. degrees from the Dept. of EE-Systems at
USC in 1995 and 1997, respectively. Currently,
he is an Adjunct Professor at the Department of
EECS at UC, Irvine. 
In the recent past, he was a Consulting
Chief Technologist at the Boeing Company 
and the CTO of TierFleet.
He is the inventor of several US patents,
has published more than seventy scholarly reviewed
articles, and authored more than twenty design articles associated with
deployed industry products. Dr. Yousefi'zadeh is/was with the editorial board of
IEEE Transaactions on Wireless Communications,
IEEE Communications Letters, IEEE Wireless Communications
Magazine, the lead guest editor of 
IEEE JSTSP the issue of April 2008, and Journal
of Communications Networks.
He was the founding Chairperson of systems' management
workgroup of the SNIA and a member of the
scientific advisory board of Integrated Media Services Center at USC. 
He is a Senior Member of the IEEE and the recipient of multiple
best paper, faculty, and engineering excellence awards.
\end{IEEEbiography}
\begin{IEEEbiography}%
[{\includegraphics[width=1in,height=1.25in,clip,keepaspectratio]{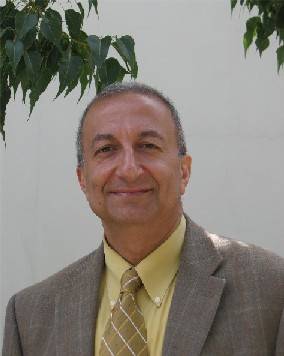}}]{Faryar Jabbari}
Faryar Jabbari (SM) received his PhD degree from UCLA in 1986. He is currently on the faculty of the Department of Mechanical and Aerospace Engineering Department of University of California, Irvine. His areas of interest are control theory and its applications, particularly to energy systems.  
He has served as an Associate Editor for the IEEE Transactions on Automatic Control and Automatica. He was the Program Chair for 2009 IEEE Conference on Decision and Control and 2011 American Control Conference.
\end{IEEEbiography}

\end{document}